\documentclass[11pt]{article}
\usepackage{amsfonts,amsmath,amssymb,boxedminipage,color,url,fullpage,nccmath,amsthm}

\usepackage{color}

\definecolor{weborange}{rgb}{.8,.3,.3}
\definecolor{webblue}{rgb}{0,0,.8}
\definecolor{internallinkcolor}{rgb}{0,.5,0}
\definecolor{externallinkcolor}{rgb}{0,0,.5}

\usepackage[pageanchor=false,
colorlinks=true,
urlcolor=externallinkcolor,
linkcolor=internallinkcolor,
filecolor=externallinkcolor,
citecolor=internallinkcolor,
pdfstartview=FitH]{hyperref}

\usepackage{enumerate,paralist}
\usepackage[numbers]{natbib} 
\usepackage[labelfont=bf]{caption}
\usepackage{aliascnt}
\usepackage{cleveref}
\usepackage{xspace}
\usepackage{ulem}
\usepackage{relsize}
\usepackage{graphics}
\usepackage[usestackEOL]{stackengine}

\usepackage[T1]{fontenc}

\newcommand{\remove}[1]{}
\newcommand{\Draft}[1]{\ifdefined\IsDraft\texttt{ #1} \fi}

\newcommand{\TLLNCS}[2]{\ifdefined\IsLLNCS#1\else #2 \fi}

\ifdefined\IsDraft
\newcommand{\authnote}[2]{{\bf [{\color{red} #1's Note:} {\color{blue} #2}]}}
\else
\newcommand{\authnote}[2]{}
\fi

\ifdefined\IsTrackChanges
\newcommand{\stkout}[1]{\ifmmode\text{\sout{\ensuremath{#1}}}\else\sout{#1}\fi}
\newcommand{\deleted}[2]{{\textbf{Deleted:}~{\color{red} \stkout{#2} }}}

\else

\newcommand{\deleted}[2]{}

\fi

\newcommand{\sdotfill}{\textcolor[rgb]{0.8,0.8,0.8}{\dotfill}} 

\newenvironment{protocol}{\begin{proto}}{\vspace{-\topsep}\sdotfill\end{proto}}
\newenvironment{algorithm}{\begin{algo}}{\vspace{-\topsep}\sdotfill\end{algo}}

\newcommand{\Ensuremath}[1]{\ensuremath{#1}\xspace}
\newcommand{\MathAlg}[1]{\mathsf{#1}}

\newcommand{\MathAlgX}[1]{\Ensuremath{\MathAlg{#1}}}

\newcommand \mycaption {\small }     
\newcommand \mylabel {}

\ifdefined\IsLLNCS
\newenvironment{nfbox}[3]{
	\renewcommand \mycaption {#1}
	\renewcommand \mylabel {#2}
	\begin{center}\small
		\begin{tabular}{|ll|}
			\hline
			\hspace{.3ex}
			\begin{minipage}{.97\linewidth}
				\vspace{0.5ex}
				#3}
			{\smallskip
				\captionof{figure}{\mycaption}
				\label{\mylabel}
			\end{minipage}
			&\hspace{.3ex} \\
			\hline
		\end{tabular}
	\end{center}    
}
\else

\fi

\newcommand{\aka} {also known as,\xspace}
\newcommand{\resp}{resp.,\xspace}
\newcommand{\ie}  {i.e.,\xspace}
\newcommand{\eg}  {e.g.,\xspace}

\newcommand{\wrt} {with respect to\xspace}
\newcommand{\wlg} {without loss of generality\xspace}
\newcommand{\Wlg} {Without loss of generality\xspace}

\newcommand{\abs}[1]{\left\lvert #1 \right\rvert}
\newcommand{\ceil}[1]{\left\lceil #1 \right\rceil}

\newcommand{\set}[1]{\ens{#1}}
\newcommand{\sset}[1]{\{#1\}}

\newcommand{\floor}[1]{\left \lfloor#1 \right \rfloor}

\newcommand{\assign}{\ensuremath{\mathrel{\vcenter{\baselineskip0.5ex \lineskiplimit0pt \hbox{\scriptsize.}\hbox{\scriptsize.}}}=}}

\newcommand{\iith}[1] {$#1$'th\xspace}

\newcommand{\jth}           {\iith{j}}

\newcommand{\rth}           {\iith{r}}

\newcommand{\half}{\tfrac{1}{2}}

\newcommand{\N}{{\mathbb{N}}}

\newcommand{\F}{{\cal F}}

\newcommand{\zo}{\{0,1\}}
\newcommand{\zn}{{\zo^n}}

\newcommand{\eps}{\varepsilon}

\newcommand{\la}{\gets}

\newcommand{\nxtmsg}{\alpha}
\newcommand{\nextmsg}{\mathsf{next\mhyphen msg}}
\newcommand{\outputf}{\mathsf{output}}
\newcommand{\stat}{\mathsf{stat}}
\newcommand{\wit}{\mathsf{wit}}

\newcommand{\negl}{\operatorname{neg}}

\newcommand{\Supp}{\operatorname{Supp}}

\newcommand{\maj}{\operatorname*{maj}}
\newcommand{\argmax}{\operatorname*{argmax}}
\newcommand{\argmin}{\operatorname*{argmin}}

\newcommand{\sk}{\mathit{sk}}
\newcommand{\vk}{\mathit{vk}}

\newcommand{\halt}{\mathsf{Halt}}

\newcommand{\class}[1]{\mathrm{#1}}

\newcommand{\NP}{\class{NP}}

\renewcommand{\cref}{\Cref}

\TLLNCS{
	\newaliascnt{claiml}{theorem}
	\newtheorem{claiml}[claiml]{Claim}
	\aliascntresetthe{claiml}
	
	\renewenvironment{claim}{\begin{claiml}}{\end{claiml}}
}
{
	\newtheorem{theorem}{Theorem}[section]

	\newaliascnt{lemma}{theorem}
	\newtheorem{lemma}[lemma]{Lemma}
	\aliascntresetthe{lemma}
	
	\newaliascnt{claim}{theorem}
	\newtheorem{claim}[claim]{Claim}
	\aliascntresetthe{claim}

	\newaliascnt{corollary}{theorem}
	
	\aliascntresetthe{corollary}
	
	\newaliascnt{proposition}{theorem}
	\newtheorem{proposition}[proposition]{Proposition}
	\aliascntresetthe{proposition}

	\newaliascnt{conjecture}{theorem}
	\newtheorem{conjecture}[conjecture]{Conjecture}
	\aliascntresetthe{conjecture}
	
	\newaliascnt{adversary}{theorem}
	
	\aliascntresetthe{adversary}

	\newaliascnt{definition}{theorem}
	\newtheorem{definition}[definition]{Definition}
	\aliascntresetthe{definition}

	\newaliascnt{remark}{theorem}
	\newtheorem{remark}[remark]{Remark}
	\aliascntresetthe{remark}

	\newaliascnt{example}{theorem}
	
	\aliascntresetthe{example}
}

\crefname{lemma}{Lemma}{Lemmas}
\crefname{figure}{Figure}{Figures}
\crefname{claim}{Claim}{Claims}
\crefname{corollary}{Corollary}{Corollaries}
\crefname{proposition}{Proposition}{Propositions}
\crefname{conjecture}{Conjecture}{Conjectures}
\crefname{definition}{Definition}{Definitions}
\crefname{remark}{Remark}{Remarks}
\crefname{exmaple}{Example}{Examples}

\newaliascnt{construction}{theorem}

\aliascntresetthe{construction}
\crefname{construction}{Construction}{Constructions}

\newaliascnt{fact}{theorem}

\aliascntresetthe{fact}
\crefname{fact}{Fact}{Facts}

\newaliascnt{notation}{theorem}

\aliascntresetthe{notation}
\crefname{notation}{Notation}{Notation}

\crefname{equation}{Equation}{Equations}

\newaliascnt{proto}{theorem}

\newtheorem{proto}[proto]{Protocol}

\aliascntresetthe{proto}
\crefname{proto}{protocol}{protocols}

\newaliascnt{algo}{theorem}
\newtheorem{algo}[algo]{Algorithm}
\aliascntresetthe{algo}
\crefname{algo}{algorithm}{algorithms}

\newaliascnt{expr}{theorem}
\newtheorem{expr}[expr]{Experiment}
\aliascntresetthe{expr}
\crefname{experiment}{experiment}{experiments}

\newaliascnt{assum}{theorem}
\newtheorem{assum}[assum]{Assumption}
\aliascntresetthe{assum}
\crefname{assumption}{assumption}{assumptions}

\newaliascnt{scen}{theorem}
\newtheorem{scen}[scen]{Scenario}
\aliascntresetthe{scen}
\crefname{scenario}{scenario}{scenarios}

\def\FullBox{$\Box$}
\def\qed{\ifmmode\qquad\FullBox\else{\unskip\nobreak\hfil
		\penalty50\hskip1em\null\nobreak\hfil\FullBox
		\parfillskip=0pt\finalhyphendemerits=0\endgraf}\fi}

\def\qedsketch{\ifmmode\Box\else{\unskip\nobreak\hfil
		\penalty50\hskip1em\null\nobreak\hfil$\Box$
		\parfillskip=0pt\finalhyphendemerits=0\endgraf}\fi}

\newcommand{\pr}[1]{\Pr\left[#1\right]}
\newcommand{\ppr}[2]{\Pr_{#1}\left[#2\right]}

\newcommand{\Ac}{\MathAlgX{A}}

\newcommand{\Pc}{\mathsf{P}}

\newcommand{\Bc}{\mathsf{B}}
\newcommand{\Cc}{\mathsf{C}}

\newcommand{\ens}[1]{\left\{#1\right\}}
\newcommand{\size}[1]{\left|#1\right|}
\newcommand{\ssize}[1]{|#1|}

\def\state{{\sf state}}
\newcommand{\ppt}{{\sc ppt}\xspace}

\newcommand{\cH}{{\cal{H}}}
\newcommand{\cA}{\mathcal{A}}
\newcommand{\cB}{\mathcal{B}}
\newcommand{\cC}{\mathcal{C}}

\newcommand{\cE}{\mathcal{E}}
\newcommand{\cF}{\mathcal{F}}

\newcommand{\cL}{\mathcal{L}}
\newcommand{\cP}{\mathcal{P}}

\newcommand{\cR}{\mathcal{R}}
\newcommand{\cS}{\mathcal{S}}
\newcommand{\cT}{\mathcal{T}}

\newcommand{\cV}{\mathcal{V}}
\newcommand{\cW}{\mathcal{W}}

\newcommand{\oS}{{\overline{\cS}}}
\newcommand{\st}{\text{ s.t.\ }}

\newcommand{\Halt}{{\cal{S}}}

\newcommand{\cs}{{\cal{S}}}

\newcommand{\xs}{x^\ast}

\newcommand{\Tableofcontents}{
	\ifdefined\IsLLNCS \else
	\thispagestyle{empty}
	\pagenumbering{gobble}
	\clearpage
	\setcounter{tocdepth}{2}
	\tableofcontents
	\thispagestyle{empty}
	\clearpage
	\pagenumbering{arabic}
	\fi
}

\newcommand{\vect}[1]{{ \boldsymbol{#1}}}

\newcommand{\vf}{\vect{f}}
\newcommand{\vm}{\vect{m}}
\newcommand{\vr}{\vect{r}}

\newcommand{\vv}{\vect{v}}
\newcommand{\vx}{\vect{x}}

\newcommand{\party}[1]{%
	\IfEqCase{#1}{%
		{1}{\Ac}
		{2}{\Bc}
		{3}{\Cc}
	}[\PackageError{\party}{Undefined option to party: #1}{}]%
}%

\newcommand{\Adv}{\Ac} 

\newcommand{\secParam}{\kappa}

\newcommand{\Party}{\MathAlgX{P}}

\newcommand{\Sim}{\MathAlgX{S}}

\mathchardef\mhyphen="2D

\newcommand{\view}{\mbox{\footnotesize {\sc view}}}

\renewcommand{\pr}[1]{\ppr{}{#1}}

\renewcommand{\sb}{\set{\Sigma \cup \bot}}
\newcommand{\sn}{\Sigma^n}

\newcommand{\rn}{\cR^n}
\newcommand{\sbn}{\sb^n}

\def\eps{\varepsilon}

\newcommand{\rbot}{\set{\cR \cup \set{\bot}}}

\newcommand{\vstar}{\mathbf{v}^{\star}}

\newcommand{\qand}{\quad \land \quad}	

\newcommand{\adaptive}{{\MathAlgX{adaptive}\xspace}}
\newcommand{\nonadaptive}{{\MathAlgX{non\mhyphen adaptive}\xspace}}

\newcommand{\BA}{\MathAlgX{BA}}
\newcommand{\Agr}{\MathAlgX{agreement}}
\newcommand{\Vld}{\MathAlgX{validity}}
\renewcommand{\Halt}{\MathAlgX{halting}}

\renewcommand{\o}{b}
\newcommand{\oo}{\overline{\o}}

\newcommand{\omin}{\set{\o,\bot}}

\newcommand{\vz}{\vv_0}
\newcommand{\vo}{\vv_1}
\newcommand{\oP}{{\overline{\cP}}}
\newcommand{\bns}{\mathbf{D}_{n,\sigma}}

\newcommand{\err}{\mathsf{err}}
\newcommand{\fnk}{\floor{(n-k)/2}}
\newcommand{\cnk}{\ceil{(n-k)/2}}
\newcommand{\fnt}{\floor{(n-t)/2}}
\newcommand{\cnt}{\ceil{(n-t)/2}}

\newcommand{\mon}{\set{\bot^n}}

\newcommand{\FirstErr}{2^{t-n}}

\title{On the Round Complexity of Randomized Byzantine Agreement\thanks{A preliminary version of this work appeared in DISC'19~\cite{CHMOS19}.}
\Draft{\\{\small \sc Working Draft: Please Do Not Distribute}}
}

\author{Ran Cohen\thanks{Efi Arazi School of Computer Science, Reichman University. E-mail: \texttt{cohenran@idc.ac.il}. Research supported in part by NSF grant no.\ 2055568. Some of this work was done while the author was a post-doc at Tel Aviv University, supported by ERC starting grant 638121.}
\and Iftach Haitner\thanks{School of Computer Science, Tel Aviv University. E-mail: \texttt{iftachh@taux.tau.ac.il}. Member of the Check Point Institute for Information Security. Research supported by Israel Science Foundation grant 666/19.}~\footnotemark[6] %
\and Nikolaos Makriyannis\thanks{Fireblocks. E-mail: \texttt{n.makriyannis@gmail.com}. This work was done while the author was a post-doc at Technion, supported by ERC advanced grant 742754.}~\footnotemark[6]
\and Matan Orland\thanks{School of Computer Science, Tel Aviv University. E-mail: \texttt{matanorland@mail.tau.ac.il}.}
\footnote{Research supported by ERC starting grant 638121.}
\and Alex Samorodnitsky\thanks{School of Engineering and Computer Science, The Hebrew University of Jerusalem.\newline{} E-mail: \texttt{salex@cs.huji.ac.il}. Research partially supported by ISF grant 1724/15.}
}

\begin{document}

\sloppy
\maketitle
\begin{abstract}
We prove lower bounds on the round complexity of \emph{randomized} Byzantine agreement (BA) protocols, bounding the halting probability of such protocols after one and two rounds. In particular, we prove that:

\begin{enumerate}
\item 
BA protocols resilient against $n/3$ [\resp $n/4$] corruptions terminate (under attack) at the end of the first round with probability at most $o(1)$ [\resp $1/2+ o(1)$].

\item
BA protocols resilient against a fraction of corruptions greater than $1/4$ terminate at the end of the second round with probability at most $1-\Theta(1)$.

\item 
For a large class of protocols (including all BA protocols used in practice) and under a plausible combinatorial conjecture,
BA protocols resilient against a fraction of corruptions greater than $1/3$ [\resp $1/4$] terminate at the end of the second round with probability at most $o(1)$ [\resp $1/2 + o(1)$].
\end{enumerate}
The above bounds hold even when the parties use a trusted setup phase, \eg a public-key infrastructure (PKI).

The third bound essentially matches the recent protocol of \citeauthor{Micali17} (ITCS'17) that tolerates up to $n/3$ corruptions and terminates at the end of the third round with constant probability.
\end{abstract}

\vfill

\Tableofcontents

\section{Introduction}\label{sec:intro}

Byzantine agreement (BA)~\cite{PSL80,LSP82} is one of the most important problems in theoretical computer science. In a BA protocol, a set of $n$ parties wish to jointly agree on one of the honest parties' input bits.
The protocol is \emph{$t$-resilient} if no set of $t$ corrupted parties can collude and prevent the honest parties from completing this task.
In the closely related problem of \emph{broadcast}, all honest parties must agree on the message sent by a (potentially corrupted) sender.
Byzantine agreement and broadcast are fundamental building blocks in distributed computing and cryptography, with applications in fault-tolerant distributed systems~\cite{CL99,KBCCEGGRWWWZ00}, secure multiparty computation~\cite{Yao82,GMW87,BGW88,CCD88}, and more recently, blockchain protocols~\cite{SM16,GHMVZ17,PS18}.

In this work, we consider the \emph{synchronous} communication model, where the protocol proceeds in rounds. It is well known that in the plain model, without any trusted setup assumptions, BA and broadcast can be solved if and only if $t<n/3$~\cite{PSL80,LSP82,FLM85,GM93}. Assuming the existence of digital signatures and a public-key infrastructure (PKI), BA can be solved in the honest-majority setting $t<n/2$, and broadcast under any number of corruptions $t<n$~\cite{DS83}. Information-theoretic variants that remain secure against computationally unbounded adversaries exist using information-theoretic pseudo-signatures~\cite{PW92}.

An important aspect of BA and broadcast protocols is their \emph{round complexity}. For deterministic $t$-resilient protocols, $t+1$ rounds are known to be sufficient~\cite{DS83,GM93} and necessary~\cite{FL82,DS83}. The breakthrough results of \citet{Ben-Or83} and \citet{Rabin83} showed that this limitation can be circumvented using randomization. In particular, \citet{Rabin83} used \emph{random beacons} (common random coins that are secret-shared among the parties in a trusted setup phase) to construct a BA protocol resilient to $t<n/4$ corruptions. The failure probability of Rabin's protocol after $r$ rounds is $2^{-r}$, and the \emph{expected} number of rounds to reach agreement is constant. This line of research culminated with the work of \citet{FM97} who showed how to compute the common coins from scratch, yielding expected-constant-round BA protocol in the plain model, resilient to $t<n/3$ corruptions. \citet{KK06} gave an analogue result in the PKI-model for the honest-majority case. Recent results used trusted setup and cryptographic assumptions to establish a surprisingly small expected round complexity, namely $9$ for $t<n/3$~\cite{Micali17} and $10$ for $t<n/2$~\cite{MV17,ADDNR19}.

The expected-constant-round protocols mentioned above are guaranteed to terminate (with negligible error probability) within a poly-logarithmic number of rounds.
The lower bounds on the guaranteed termination from~\cite{FL82,DS83} were generalized by \cite{CMS89,KY84}, showing that any randomized $r$-round protocol must fail with probability at least $(c\cdot r)^{-r}$ for some constant $c$; in particular, randomized agreement with sub-constant failure probability cannot be achieved in \emph{strictly} constant rounds. However, to date there is no lower bound on the \emph{expected} round complexity of randomized BA.

In this work, we tackle this question and show new lower bounds for randomized BA. To make the discussion more informative, we consider a more explicit definition that bounds the halting probability within a specific number of rounds. A lower bound based on such a definition readily implies a lower bound on the expected round complexity of the BA protocol.

\subsection{The Model}\label{sec:intro:model}
We start with describing in more details the model in which our lower bounds are given. In the BA protocols considered in this work, the parties are communicating over a synchronous network of private and authenticated channels. Each party starts the protocol with an input bit and upon completion decides on an output bit. The protocol is $t$-resilient if when facing $t$ colluding parties that attack the protocol it holds that: (1) all honest parties agree on the same output bit (\emph{agreement}), (2) if all honest parties start with the same input bit, then this is the common output bit (\emph{validity}), and (3) the protocol eventually terminates (\emph{termination}). The protocols might have a \emph{trusted setup phase}: a trusted external party samples correlated values (or receives a value from each party) and distributes them among the parties. A setup phase is known to be essential for tolerating $t\geq n/3$ corruptions, and seems to be crucial for highly efficient protocols such as \cite{Micali17,SM16,MV17,ADDNR19,ACDNPRS19}. The trusted setup phase is typically implemented using (heavy) secure multiparty computation \cite{BCGTV15,BGG18}, distributed key generation \cite{Pedersen91,GJKR99}, via a public-key infrastructure (see \cite{BCG21} for a discussion on different flavors of PKI), or with a random oracle (that can be used to model proofs of work)~\cite{PS17}.

\paragraph{Locally consistent adversaries.}
The attacks presented in this paper require very limited capabilities from the corrupted parties (a limitation that makes our bounds stronger). Specifically, a corrupted party can deviate from the protocol only by: (1) prematurely aborting, and (2) altering (possibly a multiple number of times) its input bit and/or incoming messages from corrupted parties (see \cref{sec:OurResults:Adv} for a precise definition). We emphasize that corrupted parties sample their random coins honestly (and use the same coins for all messages sent). In addition, they do not lie about messages received from honest parties.

\paragraph{Public-randomness protocols.}
In many randomized protocols, including all those used in practice, cryptography is merely used to provide \emph{message authentication}---preventing a party from lying about the messages it received---and \emph{verifiable randomness}---forcing the parties to toss their coins correctly. The description of such protocols can be greatly simplified if only security against locally consistent adversaries is required (in which corrupted parties do not lie about their coin tosses and their incoming messages from honest parties). This motivates the definition of \emph{public-randomness} protocols, where each party publishes its local coin tosses for each round (the party's first message also contains its setup parameter, if such exists).
Although our attacks apply to arbitrary BA protocols, we show even stronger lower bounds for public-randomness protocols.

We illustrate the simplicity of the model by considering the BA protocol of \citet{Micali17}. In this protocol, the cryptographic tools, digital signatures and verifiable random functions (VRFs),\footnote{A pseudorandom function that provides a non-interactively verifiable proof for the correctness of its output.} are used to allow the parties elect leaders and toss coins with probability $2/3$ as follows: each party $\Party_i$ in round $r$ evaluates the VRF on the pair $(i,r)$ and multicasts the result. The leader is set to be the party with the smallest VRF value, and the coin is set to be the least-significant bit of this value. Since these values are uniformly distributed $\secParam$-bit strings ($\secParam$ is the security parameter), and there are at least $2n/3$ honest parties, the success probability is $2/3$. (Indeed, with probability $1/3$, the leader is corrupted, and can send its value only to a subset of the parties, creating disagreement.)

When considering locally consistent adversaries, \citeauthor{Micali17}'s protocol can be significantly simplified by having each party randomly sample and multicast a uniformly distributed $\secParam$-bit string (cryptographic tools and setup phase are no longer needed). Corrupted parties can still send their values to a subset of honest parties as before, but they cannot send different random values to different honest parties.

A similar simplification applies to other BA protocols that are based on leader election and coin tosses such as \cite{FM97,FG03,KK06} (private channels are used for a leader-election sub-protocol), \cite{MV17,ADDNR19} (cryptography is used for coin-tossing and message-authentication), and \cite{SM16,ACDNPRS19} (cryptography is used to elect a small committee per round).\footnote{Unlike the aforementioned protocols that use ``simple'' preprocess and ``light-weight'' cryptographic tools, the protocol of \citet{Rabin83} uses a heavy, per execution, setup phase (consisting of Shamir sharing of a random coin for every potential round) that we do not know how to cast as a public-randomness protocol.}

\begin{proposition}[Malicious security to locally consistent public-randomness protocol, informal]\label{prop:mal_to_local}
Each of the BA protocols of \cite{FM97,FG03,KK06,Micali17,SM16,MV17,ADDNR19,ACDNPRS19} induces a public-randomness BA protocol secure against locally consistent adversaries, with the same parameters.
\end{proposition}

\paragraph{A useful abstraction for protocol design.}
To complete the picture, we remark that security against locally consistent adversaries, which may seem somewhat weak at first sight, can be compiled using standard cryptographic techniques into security against arbitrary adversaries. This reduction becomes lossless, efficiency-wise and security-wise, when applied to public-randomness protocols. Thus, building public-randomness protocols secure against locally consistent adversaries is a useful abstraction for protocol designers that want to use what cryptography has to offer, but without being bothered with the technical details.
See more details in \cref{sec:intro:LocalToFull}.

\paragraph{Connection to the full-information model.}
The public-randomness model can be viewed as a restricted form of the \emph{full-information model}~\cite{CC84,BL85,GGL98,BB98,BPV06,GPV06,KKKSS08,Lewko11,KS13,LL13}. In the latter model, the adversary is computationally unbounded and has complete access to all the information in the system, \ie it can listen to all transmitted messages and view the internal states of honest parties (such an adversary is also called \emph{intrusive}~\cite{CC84}). One of the motivations to study full-information protocols is to separate \emph{randomization} from \emph{cryptography} and see to what extent randomization alone can speed up Byzantine agreement. \citet{BB98} showed that any full-information BA protocol tolerating $t=\Theta(n)$ adaptive, fail-stop corruptions (\ie the adversary can dynamically choose which parties to crash) runs for $\tilde \Omega(\sqrt{n})$ rounds. \citet{GPV06} constructed an $O(\log{n})$-round BA protocol tolerating $t=(1/3-\eps)n$ static, malicious corruptions, for an arbitrarily small constant $\eps>0$.

We chose to state our results in the public-randomness model for two reasons. First, our lower bounds readily extend to lower bounds in the full-information model (since we consider weaker adversarial capabilities, \eg all our attacks are efficient). Second, when considering locally consistent adversaries, public-randomness captures essentially what efficient cryptography has to offer. Indeed, all protocol used in practice can be cast as public-randomness protocols tolerating locally consistent adversaries (\cref{prop:mal_to_local}) and every public-randomness protocol secure against locally consistent adversaries can be compiled, using cryptography, to malicious security in the standard model, where security relies on secret coins (see \cref{thm:local_to_malicious} below).

We note that it is known how to compile certain full-information protocols and ``boost'' their security from fail-stop into malicious; however, these compilers capture either deterministic protocols~\cite{Hadzilacos87,Bracha84,NT90} or protocols with a non-uniform source of randomness (namely, an SV-source~\cite{SV84})~\cite{GPV06}. It is unclear whether these compilers can be extended to capture arbitrary protocols (this is in fact stated as an open question in~\cite{Bracha84,GPV06}).
In addition, these compilers are designed to be information theoretic and not rely on cryptography; thus, they do not model highly efficient protocols used in practice.

\subsection{Our Results}\label{sec:intro:ourResult}
We present three lower bounds on the halting probability of randomized BA protocols.
To keep the following introductory discussion simple, we will assume that both validity and agreement properties hold perfectly, without error.
Throughout we consider $t<n/2$ (as otherwise Byzantine agreement cannot be achieved).

\paragraph{First-round halting.}
Our first result bounds the halting probability after a single communication round. This is the simplest case since parties cannot inform each other about inconsistencies they encounter. Indeed, the established lower bound is quite strong, showing an exponentially small bound on the halting probability when $t\geq n/3$, and exponentially close to $1/2$ when $t\geq n/4$.

\begin{theorem}[First-round halting, informal]\label{thm:intro:FirstRound}
Let $\Pi$ be an $n$-party BA protocol and let $\gamma$ denote the halting probability after a single communication round facing a locally consistent, static, adversary corrupting $t$ parties. Then,
\begin{itemize}
	\item $n/2>t \ge n/3$ implies $\gamma \le \FirstErr$ for arbitrary protocols, and $\gamma=0$ for public-randomness protocols.
	\item $n/2>t \ge n/4$ implies $\gamma \le 1/2+\FirstErr$ for arbitrary protocols, and $\gamma \leq 1/2$ for public-randomness protocols.
\end{itemize}
\end{theorem}

Note that the deterministic $(t+1)$-round, $t$-resilient BA protocol of \citet{DS83} can be cast as a locally consistent public-randomness protocol (in the plain model).\footnote{When considering locally consistent adversaries, the impossibility of BA for $t\geq n/3$ does not apply.}
\cref{thm:intro:FirstRound} shows that for $n=3$ and $t=1$, this two-round BA protocol is essentially optimal and cannot be improved via randomization (at least without considering complex protocols that cannot be cast as public-randomness protocols).

\paragraph{Second-round halting for arbitrary protocols.}
Our second result considers the halting probability after two communication rounds.
This is a much more challenging regime, as honest parties have time to detect inconsistencies in first-round messages. Our bound for arbitrary protocols in this case is weaker, and shows that when $t>n/4$, the halting probability is bounded away from $1$.

\begin{theorem}[Second-round halting, arbitrary protocols, informal]\label{thm:intro:SecondRound:Arb}
Let $\Pi$ be an $n$-party BA protocol and let $\gamma$ denote the halting probability after two communication rounds facing a locally consistent, static, adversary corrupting $t=(1/4+\eps)\cdot n$ parties.
Then, $\gamma \le 1 - (\eps /5)^2$.
\end{theorem}

\paragraph{Second-round halting for public-randomness protocols.}
\cref{thm:intro:SecondRound:Arb} bounds the second-round halting probability of arbitrary BA protocols away from one. For public-randomness protocol we achieve a much stronger bound. The attack requires \emph{adaptive} corruptions (as opposed to \emph{static} corruptions in the previous case) and is based on a combinatorial conjecture that is stated below.\footnote{The attack holds even without assuming \cref{con:intro:IsoBot} when considering \emph{strongly adaptive} corruptions~\cite{GKP15}, in which an adversary sees all messages sent by honest parties in any given round and, based on the messages' content, decides whether to corrupt a party (and alter its message or sabotage its delivery) or not. Similarly, the conjecture is not required if each party is limited to tossing a single unbiased coin. These extensions are not formally proved in this paper.\label{footnote:no_conjecture}}

\begin{theorem}[Second-round halting, public-randomness protocols, informal]\label{thm:intro:SecondRound:PR}
Let $\Pi$ be an $n$-party public-randomness BA protocol and let $\gamma$ denote the halting probability after two communication rounds facing a locally consistent adversary adaptively corrupting $t$ parties.
Then, for sufficiently large $n$ and assuming \cref{con:intro:IsoBot} holds,
\begin{itemize}
\item $t > n/3$ implies $\gamma=0$.
\item $t > n/4$ implies $\gamma \leq 1/2$.
\end{itemize}
\end{theorem}

\cref{thm:intro:SecondRound:PR} shows that for sufficiently large $n$, any public-randomness protocol tolerating $t>n/3$ locally consistent corruptions cannot halt in less than three rounds (unless \cref{con:intro:IsoBot} is false). In particular, its expected round complexity must be at least three.

To understand the meaning of this result, recall the protocol of \citet{Micali17}. As discussed above, this protocol can be cast as a public-randomness protocol tolerating $t<n/3$ adaptive locally consistent corruptions. The protocol proceeds by continuously running a three-round sub-protocol until halting, where each sub-protocol consists of a coin-tossing round, a check-halting-on-$0$ round, and a check-halting-on-$1$ round. Executing a single instance of this sub-protocol demonstrates a halting probability of $1/3$ after three rounds.
By \cref{thm:intro:SecondRound:PR}, a protocol that tolerates slightly more corruptions, \ie $(1/3 +\eps) \cdot n$, for arbitrarily small $\eps>0$, cannot halt in fewer rounds.

\paragraph{Our techniques.}
Our attacks follow the spirit of many lower bounds on the round complexity on BA and broadcast~\cite{FL82,DS83,KY84,DRS90,GKKO07,AH10}. The underlying idea is to start with a configuration in which validity assures the common output is $0$, and gradually adjust it, while retaining the same output value, into a configuration in which validity assures the common output is $1$. (For the simple case of deterministic protocols, each step of the argument requires the corrupted parties to lie about their input bits and incoming messages from other corrupted parties, but otherwise behave honestly.) Our main contribution, which departs from the aforementioned paradigm, is adding another dimension to the attack by aborting a random subset of parties (rather than simply manipulating the input and incoming messages). This change allows us to bypass a seemingly inherent barrier for this approach. We refer the reader to \cref{sec:Technique} for a detailed overview of our attacks.

We remark that a similar approach was employed by \citet{AC08} for obtaining lower bounds on consensus protocols in the asynchronous shared-memory model, a flavor of BA in a communication model very different to the one considered in the present paper. Specifically, \cite{AC08} showed that in an asynchronous shared-memory system, $\Theta(n^2)$ steps are required for $n$ processors to reach agreement when facing $\Theta(n)$ \emph{computationally unbounded strongly adaptive} corruptions (see \cref{footnote:no_conjecture}). Their adversary also aborts a subset of the parties to prevent halting; however, the difference in communication model (synchronous in our work, vs.\ asynchronous in \cite{AC08}) and the adversary's power (efficient and adaptive in our work, vs.\ computationally unbounded and strongly adaptive in \cite{AC08}) yields a very different attack and analysis (though, interestingly, both attacks boil down to different variants of isoperimetric-type inequalities).

\paragraph{The combinatorial conjecture.}
We conclude the present section by motivating and stating the combinatorial conjecture assumed in \cref{thm:intro:SecondRound:PR}, and discussing its plausibility. We believe the conjecture to be of independent interest, as it relates to topics from Boolean functions analysis such as influences of subsets of variables \cite{Odonnel14} and isoperimetric-type inequalities \cite{MosselORSS2006,MosselOS2013}. The nature of our conjecture makes the following paragraphs somewhat technical, and reading them can be postponed until after going over the description of our attack in \cref{sec:Technique}.

The analysis of our attack naturally gives rise to an isoperimetric-type inequality. For limited types of protocols, we manage to prove it using Friedgut's theorem~\cite{Friedgut98} about approximate juntas and the KKL theorem~\cite{KKL88}. For arbitrary protocols, however, we can only reduce our attack to the conjecture below.

We require the following notation before stating the conjecture. Let $\Sigma$ denote some finite set.
For $\vx\in \sn$ and $\cS \subseteq [n]$, define the vector $\bot_\cS(\vx) \in \set{\Sigma \cup \bot}^n$ by assigning all entries indexed by $\cS$ with the value $\bot$, and all other entries according to $\vx$. Finally, let $\bns$ denote the distribution induced over subsets of $[n]$ by choosing each element with probability $\sigma$ independently at random.

\def\MainConj{
For any $\sigma,\lambda >0$ there exists $\delta>0$ such that the following holds for large enough $n\in \N$: let $\Sigma$ be a finite alphabet, and let $\cA_0,\cA_1 \subseteq \sbn$ be two sets such that for both $b\in \zo$:

\begin{align*}
\ppr{\cs\gets \bns}{\ppr{\vr \gets \Sigma^n}{\vr,\bot_{\cS}(\vr) \in \cA_b} \ge \lambda } \ge 1-\delta.
\end{align*}
Then,
\begin{align*}
\ppr{\substack{\cS \gets \bns \vspace{.05in}\\ \vr\gets \Sigma^n }}{\forall b\in \zo\colon \set{\vr,\bot_{\cS}(\vr)} \cap \cA_b \neq \emptyset} \ge \delta.
\end{align*}
}

\begin{conjecture}\label{con:intro:IsoBot}
\MainConj
\end{conjecture}

\noindent
Consider two large sets $\cA_0$ and $\cA_1$ which are ``stable'' in the following sense: for both $\o\in \zo$, with probability $1-\delta$ over $\cS\la \bns$, it holds that both $\vr$ and $\bot_{\cS}(\vr)$ belong to $\cA_\o$, with probability at least $\lambda$ over $\vr$. \cref{con:intro:IsoBot} stipulates that with high probability ($\ge \delta$), the vectors $\vr$ and $\bot_{\cS}(\vr)$ lie in opposite sets (\ie one is in $\cA_0$ and the other $\cA_1$), for random $\vr$ and $\cS$. It is somewhat reminiscent of the following flavor of isoperimetric inequality: for any two large sets $\cB_0$ and $\cB_1$, taking a random element from $\cB_0$ and resampling a few coordinates, yields an element in $\cB_1$ with large probability. Less formally, one can ``move'' from one set to the other by manipulating a few coordinates~\cite{MosselORSS2006,MosselOS2013}.

A few remarks are in order. First, it suffices for our purposes to show that $\delta$ is a noticeable (\ie inverse polynomial) function of $n$, rather than independent of $n$.\footnote{We remark that it is rather easy to show that $\delta\ge 2^{-n}$, which is not good enough for our purposes.} We opted for the latter as it gives a stronger attack. Second, the conjecture holds for ``natural'' sets such as balls, \ie $\cA_0$ and $\cA_1$ are balls centered around $0^n$ and $1^n$ of constant radius,\footnote{The alphabet $\Sigma$ is not necessarily Boolean, and there are a couple of subtleties in defining balls.} and ``prefix'' sets, \ie sets of the form $\cA_\o=\o^k \times \set{\Sigma \cup \bot}^{n-k}$. Furthermore, the claim can be proven when the probabilities over $\cS$ and $\vr$ are reversed, \ie ``with probability $\lambda$ over $\vr$, it holds that both $\vr$ and $\bot_{\cS}(\vr)$ belong to $\cA_\o$ with probability at least $1-\delta$ over $\cS$'', instead of the above. Interestingly, this weaker statement boils down to the aforementioned isoperimetric-type inequality (cf.~\cite{MosselORSS2006} for the Boolean case and \cite{MosselOS2013} for the non-Boolean case).

We conclude by pointing out that, as mentioned in \cref{footnote:no_conjecture}, the conjecture is not needed for certain limited cases that are not addressed in detail in the present paper. One such case is sketched out in \cref{sec:Technique}.

\subsection{Locally Consistent Security to Malicious Security}\label{sec:intro:LocalToFull}

As briefly mentioned in \cref{sec:intro:model}, protocols that are secure against locally consistent adversaries can be compiled to tolerate arbitrary malicious adversaries.
The compiler requires a PKI setup for digital signatures, verifiable random functions (VRFs)~\cite{MRV99}, and non-interactive zero-knowledge proofs (NIZK)~\cite{BFM88}. A VRF is a pseudorandom function with an additional property: using the secret key and an input $x$, the VRF outputs a pseudorandom value $y$ along with a proof string $\pi$; using the public key, everyone can use $\pi$ to verify whether $y$ is the output of $x$. We consider a trusted setup phase for establishing the PKI, where a trusted party generates VRF and signature keys for every party, securely gives the secret keys to each party, and publishes the public keys to all.

Given a protocol that is secure against locally consistent adversaries, the compiled protocol proceeds as follows, round by round.
Each party $\Party_i$ sets its random coins for the \rth round $\rho_i^r$ (together with a proof $\pi_i^r$) by evaluating the VRF over the pair $(i,r)$.
Next, for every $j\in[n]$, party $\Party_i$ uses these coins to compute the message $m^r_{i\to j}$ for $\Party_j$, signs $m^r_{i\to j}$ as $\sigma^r_{i\to j}$, and sends $(m^r_{i\to j},\sigma^r_{i\to j},\pi_i^r)$ to $\Party_j$.
Finally, $\Party_i$ sends to $\Party_j$ a NIZK proof that:
\begin{enumerate}
    \item
    There exist an input bit $b$, random coins $\rho_i^r$, as well as random coins $\rho^{r'}_i$ and incoming messages and $(m^{r'}_{1\to i},\ldots,m^{r'}_{n\to i})$ for every prior round $r'<r$, such that: (1) $\pi_i^r$ verifies that $\rho_i^r$ is the VRF output of $(i,r)$ (using the VRF public key of $\Party_i$), (2) the message $m^r_{i\to j}$ was signed by $\Party_i$, and (3) the message $m^r_{i\to j}$ is the output of the next-message function of $\Party_i$ when applied to these values.
    \item
    For $r>1$, the messages $(m^{r'}_{k\to i},\sigma^{r'}_{k\to i},\pi_k^{r'})$ received by $\Party_i$ from every $\Party_k$ in prior rounds are proven to be properly generated. That is, $\Party_k$ provided a NIZK proof that explains how $m^{r'}_{k\to i}$ was generated using random coins computed via the VRF on $(k,r')$ and on incoming messages that were signed by the senders.
\end{enumerate}

When considering public-randomness protocols, the above compilation can be made much more efficient. Instead of proving in zero knowledge the consistency of each message, each party $\Party_i$ concatenates to each message all of its incoming messages from the previous round. A receiver can now locally verify the coins used by $\Party_i$ are the VRF output of $(i,r)$ (as assured by the VRF), that the incoming messages are properly signed, and that the message is correctly generated from the internal state of $\Party_i$ (which is now visible and verified).

\begin{theorem}[Locally consistent to malicious security, folklore, informal]\label{thm:local_to_malicious}
Assume PKI for digital signatures, VRF, and NIZK. Then, an expected-constant-round BA protocol secure against locally consistent adversaries can be compiled into a maliciously secure protocol with the same parameters.
\end{theorem}

The proof of \cref{thm:local_to_malicious} can be found in \cref{sec:LocalToFull}.

\subsection{Additional Related Work}\label{sec:relatedWork}

Following the work of \citet{FM97} in the two-thirds majority setting, \citet{KK06} improved the expected round complexity to $23$, and \citet{Micali17} to $9$. In the honest-majority setting, \citet{FG03} showed expected-constant-round protocol and \citet{KK06} expected $56$ rounds. \citet{MV17} adjusted the technique from \cite{Micali17} to the honest-majority case. \citet{ADDNR19} achieved expected $10$ rounds assuming static corruptions and expected $16$ rounds assuming adaptive corruptions. \citet{ACDNPRS19} constructed an expected-constant-round protocol tolerating $(1/2-\epsilon)\cdot n$ adaptive corruptions with sublinear communication complexity. In the dishonest-majority setting, \citet{GKKO07} constructed a broadcast protocol with expected $O(k^2)$ rounds, tolerating $t<n/2+k$ corruptions, that was improved by \citet{FN09} to expected $O(k)$ rounds.

\citet{AH10} extended the results of \citet{CMS89} and of \citet{KY84} on guaranteed termination of randomized BA protocols to the asynchronous setting, and provided a tight lower bound.

Randomized protocols with expected constant round complexity have \emph{probabilistic termination}, which requires delicate care \wrt composition (\ie their usage as subroutines by higher-level protocols). Parallel composition of randomized BA protocols was analyzed in \cite{Ben-Or83,FG03}, sequential composition in \cite{LLR06}, and universal composition in \cite{CCGZ16,CCGZ17}.

\subsection{Open Questions}\label{sec:OpenQuest}
Our attack on two-round halting of public-randomness protocols is based on \cref{con:intro:IsoBot}. In this work we prove special cases of this conjecture, but proving the general case remains an open challenge.

A different interesting direction is to bound the halting probability of protocols when $t<n/4$. It is not clear how to extend our attacks to this regime.

\subsection*{Paper Organization}

In \cref{sec:Technique} we present a technical overview of our attacks. The formal model and the exact bounds are stated in \cref{sec:OurResult}. The proof of the first-round halting is given in \cref{sec:FirstRound}, and for second-round halting in \cref{sec:SecondRound}.
The proof of \cref{thm:local_to_malicious} appears in \cref{sec:LocalToFull}.

\newcommand{\wb}[1]{\overline{#1}}
\newcommand{\rr}{\mathbf{r}}
\newcommand{\cN}{\mathcal{N}}
\newcommand{\cK}{\mathcal{K}}
\newcommand{\ham}{\mathrm{dist}}
\newcommand{\oh}{\overline{\cH}}
\newcommand{\oC}{\overline{\cC}}

\section{Our Techniques}\label{sec:Technique}
In this section, we outline our techniques for proving our results. We start with explaining our bound for first-round halting of arbitrary protocols (\cref{thm:intro:FirstRound}). We then move to second-round halting, starting with the weaker bound for arbitrary protocols (\cref{thm:intro:SecondRound:Arb}), and then move to the much stronger bound for public-randomness protocols (\cref{thm:intro:SecondRound:PR}).

\paragraph{Notations.}
We use calligraphic letters to denote sets, uppercase for random variables, lowercase for values, boldface for vectors, and sans-serif (\eg \Ac) for algorithms (\ie Turing Machines).
For $n\in\N$, let $[n]=\set{1,\cdots,n}$ and $(n)=\set{0,1,\cdots,n}$. Let $\ham(x, y)$ denote the hamming distance between $x$ and $y$. For a set $\cS \subseteq [n]$ let $\oS = [n]\setminus \cS$. For a set $\cR\subseteq \zo^n$, let $\cR|_{\cS}=\sset{\vx_\cS \in \zo^{\size{\cS}}\st \vx\in \cR}$, \ie $\cR|_{\cS}$ is the projection of $\cR$ on the index-set $\cS$.

Fix an $n$-party randomized BA protocol $\Pi = (\Pc_1,\ldots,\Pc_n)$. For presentation purposes, we assume that $n$ is divisible by $3$, that validity and agreement hold \emph{perfectly}, and consider no setup parameters (in the subsequent sections, we remove these assumptions). Furthermore, we only address here the case where the security threshold is $t>n/3$. The case $t>n/4$ requires an additional generic step that we defer to the technical sections of the paper. We denote by $\Pi(\vv;\vr)$ the output of an honest execution of $\Pi$ on input $\vv \in \zn$ and randomness $\vr$ (each party $\Party_i$ holds input $v_i$ and randomness $r_i$). We let $\Pi(\vv)$ denote the resulting random variable determined by the parties' random coins, and we write $\Pi(\vv) =\o$ to denote the event that the parties output $\o$ in an honest execution of $\Pi$ on input $\vv$. All corrupt parties described below are locally consistent (see \cref{sec:intro:model}).

\subsection{First-Round Halting}\label{sec:technique:1}
Assume the honest parties of $\Pi$ halt at the end of the first round with probability $\gamma>0$ when facing $t$ corruptions (on every input).
Our goal is to upperbound the value of $\gamma$. Our approach is inspired by the analogous lower-bound for deterministic protocols (see \cite{FL82,DS83}). Namely, we start with a configuration in which validity assures the common output is $0$, and, while maintaining the same output, we gradually adjust it into a configuration in which validity assures the common output is $1$, thus obtaining a contradiction. For randomized protocols, the challenge is to maintain the invariant of the output, even when the probability of halting is far from $1$. We make the following observations:

\begin{align} \label{eq:maj}
&\text{Almost pre-agreement:} \quad \ham(\vv,\o^n) \le t \implies \Pi(\vv)=\o.
\end{align}
That is, in an honest execution of $\Pi$, if the parties almost start with preagreement, \ie with at least $n-t$ of $\o$'s in the input vector, then the parties output $\o$ with probability $1$. \cref{eq:maj} follows from \emph{agreement} and \emph{validity} by considering an adversary corrupting exactly those parties with input $v_i\neq b$, and otherwise not deviating from the protocol.

\begin{align}\label{eq:nei}
&\text{Neighboring executions (N1):} \quad \ham(\vz,\vo)\le t \implies \ppr{\vr}{\Pi(\vz;\vr)=\Pi(\vo;\vr)}\ge \gamma.
\end{align}
That is, for two input vectors that are at most $t$-far (\ie the resiliency threshold), the probability that the executions on these vectors yield the same output when using the same randomness is bounded below by the halting probability. To see why \cref{eq:nei} holds, consider the following adversary corrupting subset $\cC$, for $\cC$ being the set of indices where $\vz$ and $\vo$ disagree. For an arbitrary partition $\sset{\oC_0,\oC_1}$ of $\oC$, the adversary instructs $\cC$ to send messages according to $\vz$ to $\oC_0$ and according to $\vo$ to $\oC_1$, respectively. With probability at least $\gamma$, all parties halt at the first round, and, by perfect agreement, all parties compute the same output.\footnote{In the above, we have chosen to ignore a crucial subtlety. In an execution of the protocol, it may be the case that there is a suitable message (according to $\vz$ or $\vo$) to prevent halting, yet the adversary cannot determine which one to send. In further sections, we address this issue by taking a random partition of $\oC$ (rather than an arbitrary one). By doing so, we introduce an error-term of $1/2^{n-t}$ when we upper bound the halting probability $\gamma$.} Since parties in $\oC_\o$ cannot distinguish this execution from a halting execution of $\Pi(\vv_\o;\vr)$, \cref{eq:nei} follows.

We deduce that if there are more than $n/3$ corrupt parties, then the halting probability is $0$; this follows by combining the two observations above for $\vz=0^{2n/3}1^{n/3}$ and $\vo=0^{n/3}1^{2n/3}$. Namely, by \cref{eq:maj}, it holds that $\ppr{\vr}{\Pi(\vz;\vr)=\Pi(\vo;\vr)}=0$. Thus, by \cref{eq:nei}, $\gamma=0$.

\subsection{Second-Round Halting -- Arbitrary Protocols}\label{sec:technique:2}

We proceed to explain our bound for second-round halting of arbitrary protocols. Assume the honest parties of $\Pi$ halt at the end of the second round with probability $\gamma>0$ when facing $t$ corruptions (on every input). Let $t=(1/3+\eps)\cdot n$, for an arbitrary small constant $\eps>0$. In spirit, the attack follows the footsteps of the single-round case described above; we show that neighboring executions compute the same output with good enough probability (related to the halting probability), and lower-bound the latter using the \emph{almost pre-agreement} observation. There is, however, a crucial difference between the first-round and second-round cases; the honest parties can use the second round to detect whether (some) parties are sending inconsistent messages. Thus, the second round of the protocol can be used to ``catch-and-discard'' parties that are pretending to have different inputs to different parties, and so our previous attack breaks down. (In the one-round case, we exploit the fact that the honest parties cannot verify the consistency of the messages they received.) Still, we show that there is a suitable variant of the attack that violates the agreement of any ``too-good'' scheme.

At a very high level, the idea for proving the \emph{neighboring} property is to \emph{gradually} increase the set of honest parties towards which the adversary behaves according to $\vo$ (for the remainder it behaves according to $\vz$, which is a decreasing set of parties). While the honest parties might identify the attacking parties and discard their messages, they should still agree on the output and halt at the conclusion of the second round with high probability. We exploit this fact to show that at the two extremes (where the adversary is merely playing honestly according to $\vz$ and $\vo$, respectively), the honest parties behave essentially the same. Therefore, if at one extreme (for $\vz$) the honest parties output $\o$, it follows that they also output $\o$ at the other extreme (for $\vo$), which proves the \emph{neighboring} property for the second-round case.

We implement the above by augmenting the one-round attack as follows. In addition to corrupting a set of parties that feign different inputs to different parties, the adversary corrupts an extra set of parties that is inconsistent with regards to the messages it received from the first set of corrupted parties. To distinguish between the two sets of corrupted parties, the former (first) will be referred to as ``pivot'' parties (since they pivot their input) and will be denoted $\cP$, and the latter will be referred to as ``propagating'' parties (since they carefully choose what message to propagate at the second round) and will be denoted $\cL$. We emphasize that the propagating parties deviate from the protocol only at the second round and only with regards to the messages received by the pivot parties (not with regards to their input -- as is the case for the pivot parties). In more detail, we partition $\oP = [n]\setminus \cP$ into $\ell=\lceil 1/\eps\rceil$ sets $\sset{\cL_1,\ldots,\cL_\ell}$, and we show that, unless there exists $i$ such that parties in $\cC= \cP\cup \cL_i$ violate agreement (explained below), the following must hold for neighboring executions.

\begin{align} \label{eq:nei2}
&\text{Neighbouring executions (N2):} \quad \ham(\vz,\vo)\le n/3 \implies \\
& \pr{\Pi(\vz)=\o\text{ in two rounds}} \ge \pr{ \Pi(\vo)=\o\text{ in two rounds}} - 2 (\ell +1)^2 \cdot (1-\gamma). \nonumber
\end{align}
That is, for two input vectors that are at most $n/3$--far, the difference in probability that two distinct executions (for each input vector) yield the same output within two rounds is roughly upper-bounded by the quantity $(1-\gamma)/\eps^2$ (\ie non-halting probability divided by $\eps^2$). To see that \cref{eq:nei2} holds true, fix $\vz,\vo\in \zn$ of hamming distance at most $n/3$, and let $\cP$ be the set of indices where $\vz$ and $\vo$ differ. Consider the following $\ell+1$ distinct variants of $\Pi$, denoted $\set{\Pi_0,\ldots, \Pi_\ell}$; in protocol $\Pi_i$, parties in $\cP$ send messages to $\cL_1,\ldots, \cL_i$ according to the input prescribed by $\vo$ and to $\cL_{i+1},\ldots, \cL_\ell$ according to the input prescribed by $\vz$, respectively. All other parties follow the instructions of $\Pi$ for input $\vz$. We write $\Pi_i=\o$ to denote the event that the parties not in $\cP$ output $\o$. Notice that the endpoint executions $\Pi_0$ and $\Pi_{\ell}$ are identical to honest executions with input $\vz$ and $\vo$, respectively. Let $\halt_i$ denote the event that the parties not in $\cP$ halt at the second round in an execution of $\Pi_{i}$. We point out that $\pr{ \neg \halt_i}\le (\ell+1)\cdot (1-\gamma)$, since otherwise the adversary corrupting $\cP$ and running $\Pi_i$, for a random $i\in (\ell) \assign \set{0,\ldots,\ell}$, prevents halting with probability greater than $1-\gamma$. Next, we inductively show that
\begin{align}
\pr{\Pi_i=\o \land \halt_i} \ge \pr{\Pi_{0} =\o \land \halt_0} - 2i\cdot (\ell+1) \cdot (1-\gamma),\label{eq:teke}
\end{align}
for every $i\in (\ell)$, which yields the desired expression for $i=\ell$. In pursuit of contradiction, assume \cref{eq:teke} does not hold, and let $i$ denote the smallest index for which it does not hold (observe that $i\ne 0$, by definition). Notice that
\begin{align*}
& \hspace*{-2cm}\pr{(\Pi_{i-1}= \o \land \halt_{i-1}) \land (\Pi_{i }\neq \o\land \halt_i)}\\
&\ge \pr{\Pi_{i-1}=\o \land \halt_{i-1}}-\pr{ \Pi_{i }= \o\vee \neg \halt_i}\\
&\ge \pr{\Pi_{i-1}=\o \land \halt_{i-1}}-\pr{ \Pi_{i }= \o\land \halt_i} - \pr{ \neg \halt_i}\\
&> 2\cdot (\ell+1) \cdot (1-\gamma) - \pr{ \neg \halt_i} \\
&\ge (\ell+1) \cdot (1-\gamma) > 0.
\end{align*}
The second inequality follows from union bound and $A\lor \neg B\equiv (A\land B) \lor \neg B$, the third inequality is by induction hypothesis, and the last inequality by the bound $\pr{ \neg \halt_i}\le (\ell+1)\cdot (1-\gamma)$.

It follows that an adversary corrupting $\cC=\cP \cup \cL_i$ causes disagreement with non-zero probability by acting as follows: parties in $\cP$ and $\cL_i$ send messages according to $\Pi_i$ and $\Pi_{i-1}$ to $\oC_0$ and $\oC_1$, respectively, where $\sset{\oC_0,\oC_1}$ is an arbitrary partition of $\oC= [n]\setminus \cP\cup \cL_i$. Since disagreement is ruled out by assumption, we deduce \cref{eq:teke,eq:nei2}. To conclude, we combine the \emph{almost pre-agreement} property (\cref{eq:maj}) with the \emph{neighboring} property (\cref{eq:nei2}) with $\vz=0^{2n/3}1^{n/3}$, $\vo=0^{n/3}1^{2n/3}$, and $\o=1$.
Namely, $\pr{\Pi(\vz)=1\text{ in two rounds}}=0$, by \emph{almost pre-agreement} and $\pr{\Pi(\vo)=1\text{ in two rounds}}\ge\gamma$, by \emph{almost pre-agreement} and \emph{halting}. It follows that $0 \ge \gamma - 2 (\ell +1)^2 \cdot (1-\gamma)$, by \cref{eq:nei2}, and thus $1-\frac{1}{2(\ell+1)^2+1} \ge \gamma$, which yields the desired expression.

\subsection{Second-Round Halting -- Public-Randomness Protocols}\label{sec:technique:3}

In \cref{sec:technique:2}, we ruled out ``very good'' second-round halting for arbitrary protocols via an efficient locally consistent attack. Recall that if the halting probability is close to 1, then there is a somewhat simple attack that violates agreement and/or validity. In this subsection, we discuss ruling out \emph{any} second-round halting, \ie halting probability bounded away from zero, for public-randomness protocols.

We first explain why the attack -- as is -- does not rule out second-round halting. Suppose that at the first round the parties of $\Pi$ send a deterministic function of their input, and at the second round they send the messages they received at the first round together with a uniform random bit. On input $\vv$ and randomness $\vr$, the parties are instructed \emph{not} to halt at the second round (\ie carry on beyond the second round until they reach agreement with validity) if a super-majority ($\ge n-t$) of the $v_i$'s are in agreement and $\maj(r_1,\ldots, r_n)\neq \maj(v_1,\ldots, v_n)$, \ie the majority of the random bits does not agree with the super-majority of the inputs. In all other cases, the parties are instructed to output $\maj(r_1,\ldots, r_n)$.
It is not hard to see that this protocol will halt with probability $1/2$, even in the presence of the previous locally consistent adversary (regardless of the choice of propagating parties $\cL_i$). More generally, if the randomness uniquely determines the output, then the protocol designer ensures that halting does not result in disagreement (by partitioning the randomness appropriately), and thus foiling the previous attack.\footnote{In \cref{sec:technique:2}, halting was close to $1$ and thus the randomness was necessarily ambiguous regarding the output.}

To overcome the above apparent obstacle, we introduce another dimension to our locally consistent attack; we instruct an extra set of corrupted parties to abort at the second round without sending their second-round messages. By utilizing aborting parties, the adversary can potentially decouple the output/halting from the parties' randomness and thus either prevent halting or cause disagreement.
In \cref{sec:technique:3:1}, we explain how to rule out second-round halting for a rather unrealistic class of public-randomness protocol. What makes the class of protocols unrealistic is that we assume security holds against unbounded locally consistent adversaries, and the protocol prescribes only a single bit of randomness per party per round. That being said, this case illustrates nicely our attack, and it also makes an interesting connection to Boolean functions analysis (namely, the KKL theorem~\cite{KKL88}). For general public-randomness protocols, we only know how to analyze the aforementioned attack assuming \cref{con:intro:IsoBot}, as explained in \cref{sec:technique:3:2}.

\subsubsection{``Superb'' Single-Coin Protocols}\label{sec:technique:3:1}
A BA protocol $\Pi$ is $t$-\emph{superb} if agreement and validity hold perfectly against an adaptive \emph{unbounded} locally consistent adversary corrupting at most $t$ parties, \ie the probability that such an adversary violates agreement or validity is $0$. A public-randomness protocol is \emph{single-coin}, if, at any given round, each party samples a single unbiased bit.

\begin{theorem}[Second-round halting, superb single-coin protocols]\label{bound:KKL}
For every $\eps> 0$ there exists $c>0$ such that the following holds for large enough $n$. For $t=(1/3+\eps) \cdot n$, let $\Pi$ be a $t$-superb, single-coin, $n$-party public-randomness Byzantine agreement protocol and let $\gamma$ denote the probability that the protocol halts in the second round under a locally consistent attack. Then, $\gamma\le n^{-c}$.
\end{theorem}

We assume for simplicity that the parties do not sample any randomness at the first round, and write $\vr\in \zn$ for the vector of bits sampled by the parties at the second round, \ie $r_i$ is a uniform random bit sampled by $\Party_i$.

As discussed above, our attack uses an additional set of corrupted parties of size $\sigma\cdot n$, dubbed the ``aborting'' parties and denoted $\cS$, that abort indiscriminately at the second round (the value of $\sigma$ is set to $\eps/4$ and $\ell=2\cdot \lceil1/\eps\rceil$ to accommodate for the new set of corrupted parties, \ie $\size{\cL_i}\le n\cdot \eps /2$). In more detail, analogously to the previous analysis, we consider $(\ell+1)\cdot \binom{n}{\sigma n}$ distinct variants of $\Pi$, denoted $\sset{\Pi^{\cS}_i}_{i,\cS}$ and indexed by $i\in (\ell)$ and $\cS \subseteq[n]$ of size $\sigma n$, as follows. In protocol $\Pi_i^\cS$, parties in $\cP$ send messages to $\cL_1,\ldots, \cL_i$ according to the input prescribed by $\vo$, and to $\cL_{i+1},\ldots, \cL_\ell$ according to the input prescribed by $\vz$ (recall that $\cP$ consists of exactly those indices where $\vz$ and $\vo$ differ). Parties in $\cS$ act according to $\cP$ or $\cL_j$, for the relevant $j$, except that they abort at the second round without sending their second-round messages. We write $\Pi^{\cS}_i(\vr)=\o$ to denote the event that the parties not in $\cP\cup \cS$ output $\o$, where the parties' second-round randomness is equal to $\vr$. Let $\halt^{\cS}_i$ denote the event that all parties not in $\cP\cup \cS$ halt at the second round in an execution of $\Pi^{\cS}_i$, and define $\cR_{i }^{\cS}(\o)=\sset{\vr\in \zn\st \Pi^{\cS}_i(\vr)=\o \land \halt^{\cS}_i}$. The following holds:

\begin{align}\label{eq:nei3}
&\text{Neighbouring executions (N2$\dagger$):}\\
&\quad \forall \vz,\vo \in \zn \text{ with } \ham(\vz,\vo)\le n/3, \quad \forall \o\in \zo,i\in [\ell] \assign \set{1,\ldots,\ell}\colon\nonumber\\
&\qquad\qquad \left (\forall \cS\colon \pr{\Pi^\cS_{i-1}=\o \land \halt^{\cS}_{i-1}} \ge \gamma/2 \right) \implies \left(\forall \cS\colon\pr{\Pi^\cS_{i}=\o \land \halt^{\cS}_{i}}\ge \gamma/2\right). \nonumber
\end{align}

\noindent
In words, for both $\o\in \zo$: if $\Pi_{i-1}^\cS=\o$ and halts in two rounds with large probability ($\ge \gamma/2$), for every $\cS$, then $\Pi_i^\cS=\o$ and halts in two rounds with large probability, for every $\cS$. Before proving \cref{eq:nei3}, we show how to use it to derive \cref{bound:KKL}. We apply \cref{eq:nei3} for
$\vz=0^{2n/3}1^{n/3}$, $\vo=0^{n/3}1^{2n/3}$, $\o=0$, and $i=\ell$, in combination with the properties of \emph{validity} and \emph{almost pre-agreement} (\cref{eq:maj}). Namely, by these properties, a random execution of $\Pi$ on input $\vz$ where the parties in $\cS$ abort at the second round yields output $0$ with probability at least $\gamma/2$, for every $\cS\in {\binom{[n]}{\sigma n}}$.
Therefore, by \cref{eq:nei3}, we deduce that a random execution of $\Pi$ on input $\vo$ where the parties in $\cS$ abort at the second round yields output $0$ with probability at least $\gamma/2$, for every $\cS\in {\binom{[n]}{\sigma n}}$. The latter violates either \emph{validity} or \emph{almost pre-agreement} -- contradiction.
To conclude the proof of \cref{bound:KKL}, we prove \cref{eq:nei3} by using the following corollary of the seminal KKL theorem \cite{KKL88} from \citet{BKK14}. (Recall that $\cR|_{\wb{\cS}}$ is the projection of $\cR$ on the index-set $\wb{\cS}$.)

\begin{lemma}\label{lem:KKL}
For every $\sigma, \delta\in (0,1)$, there exists $c>0$ s.t.\ the following holds for large enough~$n$. Let $\cR \subseteq \zn$ be s.t.\ $\ssize{\cR|_{\wb{\cS}}}\le (1-\delta)\cdot 2^{(1-\sigma)n}$, for every $\cS\subseteq [n]$ of size $\sigma n$. Then, $\size{\cR}\le n^{-c}\cdot 2^n$.
\end{lemma}

Loosely speaking, \cref{lem:KKL} states that for a set $\cR\subseteq \zn$, if the size of every projection on a constant fraction of indices is bounded away from one (in relative size), then the size of $\cR$ is vanishingly small
(again, in relative size).\footnote{In the jargon of Boolean functions analysis, since every large set has a $o(n)$-size index-set of influence almost one, it follows that some projection on a constant fraction of indices is almost full.}

Going back to the proof, in pursuit of contradiction, let $i\ge 1$ denote the smallest index for which \cref{eq:nei3} does not hold, and without loss of generality suppose $b=0$, \ie there exists $\cS$ such that $\ssize{\cR_{i}^{\cS}(0)}< \gamma/2 \cdot 2^{n}$, and $\ssize{\cR_{i-1}^{\cS'}(0)}\ge \gamma/2 \cdot 2^n$, for every relevant $\cS'$. We prove \cref{eq:nei3} by proving \cref{eq:halt,eq:proj}, which result in contradiction via \cref{lem:KKL}.

\begin{align}
\text{Halting:}&\qquad \ssize{\cR_{i }^{\cS}(1)}\ge \gamma/2 \cdot 2^{n} \label{eq:halt} \\
\text{Perfect agreement:}&\qquad \forall \cS' \colon \quad \ssize{\cR_{i}^{\cS}(1)|_{\wb{\cS}'}}\le (1-\gamma/2)\cdot 2^{(1-\sigma) n} \label{eq:proj}
\end{align}

\noindent
\cref{eq:halt} follows by the \emph{halting} property of $\Pi_{i}^\cS$, since the execution halts if and only if $\vr\in \cR_{i }^{\cS}(1)\cup \cR_{i }^{\cS}(0)$, and, by assumption, $\ssize{\cR_{i }^{\cS}(0)}< \gamma/2 \cdot 2^{n}$.
To conclude, we prove \cref{eq:proj} by observing that for every $\cS'$ and $b\in \zo$, and every $\vr$ and $\vr'$, if $\vr\in \cR_{i-1}^{\cS'}(0)$ and $\vr|_{\wb{\cS}'}=\vr'|_{\wb{\cS}'}$, then $\vr'\in \cR_{i-1}^{\cS'}(0)$ (by definition), \ie membership to $\cR_{i-1}^{\cS'}(0)$ does not depend on the indices of $\cS'$.
Therefore, if $\vr\in \cR_i^\cS(1)$ and $\vr|_{\wb{\cS}'}\in \cR_{i-1}^{\cS'}(0)|_{\wb{\cS}'}$, for some $\cS'$ and $\vr$, then $\vr\in \cR_{i-1}^{\cS'}(0) \cap \cR_{i}^{\cS}(1)$ which gives rise to the following attack.
The attacker controls $\cP$, $\cL_{i}$, $\cS$, and $\cS'$, and sends messages according to $\Pi_{i}^{\cS}$ and $\Pi_{i-1}^{\cS'}$ to $\oC_0$ and $\oC_1$, respectively, where $\sset{\oC_0,\oC_1}$ is an arbitrary partition of $\oC=[n]\setminus \cP\cup \cL_{i} \cup \cS\cup \cS'$. It is not hard to see the attacker violates agreement, whenever the randomness lands on $\vr$.

Finally, since $\ssize{\cR_{i-1}^{\cS'}(0)}\ge \gamma/2 \cdot 2^{n}$, we observe that $\ssize{\cR_{i-1}^{\cS'}(0)|_{\wb{\cS}'}}\ge \gamma/2 \cdot 2^{(1-\sigma) n}$, and, since $\cR_{i-1}^{\cS'}(0)|_{\wb{\cS}'}$ and $\cR_{i}^{\cS}(1)|_{\wb{\cS}'}$ are non-intersecting for every $\cS'$, it follows that $\ssize{\cR_{i}^{\cS}(1)|_{\wb{\cS}'}}\le (1-\gamma/2)\cdot 2^{(1-\sigma) n}$, which yields \cref{eq:proj}.

\begin{remark}
For superb, single-coin, public-randomness protocol, repeated application of \cref{eq:nei,lem:KKL} rules out second-round halting for arbitrary (constant) fraction of corrupted parties (and not only $n/3$ fraction).
\end{remark}

\subsubsection{General (Public-Randomness) Protocols}\label{sec:technique:3:2}
The analysis above crucially relies on the superb properties of the protocol. While it can be generalized for protocols with near-perfect statistical security and constant-bit randomness, we only manage to analyze the most general case (\ie protocols with non-perfect computational security and arbitrary-size randomness) assuming \cref{con:intro:IsoBot}. Very roughly (and somewhat inaccurately), when applying the above attack on general public-randomness protocols, the following happens for some $\delta>0$ and both values of $\o\in \zo$: for $(1-\delta)$-fraction of possible aborting subsets $\cS$, the probability that the honest parties halt in two rounds and output the same value $\o$, whether parties in $\cS$ all abort or not, is bounded below by the halting probability. Assuming \cref{con:intro:IsoBot}, it follows that with probability $\delta$ over the randomness and $\cS$, the honest parties under the attack output opposite values depending whether the parties in $\cs$ abort or not. We conclude that the agreement of the protocol is at most $\delta$. We refer the reader to \cref{sec:TwoRoundProtcol:PR} for the full details.

\section{Our Lower Bounds}\label{sec:OurResult}

In this section, we formally state our lower bounds on the round complexity of Byzantine agreement protocols. The communication and adversarial models as well as the notion of Byzantine agreement protocols we consider are given in \cref{sec:OurResults:Model}, and our bounds are formally stated in \cref{sec:OurResult:Bounds}.

\subsection{The Model}\label{sec:OurResults:Model}

\subsubsection{Protocols}
All protocols considered in this paper are \ppt (probabilistic polynomial time): the running time of every party is polynomial in the (common) security parameter (given as a unary string). We only consider Boolean-input Boolean-output protocols: apart from the common security parameter, all parties have a single input bit, and each of the honest parties outputs a single bit. For an $n$-party protocol $\Pi$, an input vector $\vv\in \zn$ and randomness $\vr$, let $\Pi(\vv; \vr)$ denote the output vector of the parties in an (honest) execution with party $\Party_i$'s input being $\vv_i$ and randomness $\vr_i$.
For a set of parties $\cP \subseteq [n]$, we denote by $\Pi(\vv; \vr)_\cP$ the output vector of the parties in $\cP$.

The protocols we consider might have a \emph{setup phase} in which before interaction starts a trusted party distributes (correlated) values between the parties. We only require the security to hold for a \emph{single} use of the setup parameters, \ie for a single instance of the BA protocol (in reality, these parameters are set once and then used for many interactions). This, however, only makes our lower bound stronger.

The communication model is \emph{synchronous}, meaning that the protocols proceed in rounds. In each round every party can send a message to every other party over a private and authenticated channel. (Allowing the protocol to be executed over private channels makes our lower bounds stronger.) It is guaranteed that all of the messages that are sent in a round will arrive at their destinations by the end of that round.

\subsubsection{Adversarial Model}\label{sec:OurResults:Adv}

We consider both \adaptive and \nonadaptive (\aka static) adversaries. An \adaptive adversary can choose which parties to corrupt for the next round immediately after the conclusion of the previous round but before seeing the next round's messages. If a party has been corrupted then it is considered corrupt for the rest of the execution. A \nonadaptive (static) adversary chooses which parties to corrupt \emph{before} the execution of the protocol begins (\ie before the setup phase, if such exists). We measure the success probability of the latter adversaries as the expectation over their choice of corrupted parties.

We consider both \emph{rushing} and \emph{non-rushing} adversaries. A non-rushing adversary chooses the corrupted parties' messages in a given round based on the messages sent in the \emph{previous} rounds. In contrast, a rushing adversary can base the corrupted parties' messages on the messages sent in the previous rounds, and on those sent by the honest parties in the \emph{current} round.

\paragraph{Locally consistent adversaries.}
As discussed in \cref{sec:intro:model}, our attack requires very limited capabilities from each corrupted party: to prematurely abort, and to lie about its input bit and incoming messages from other corrupted parties. In particular, a corrupted party tosses its local coins honestly and does not lie about incoming messages from honest parties. We now present the formal definition.

\begin{definition}[locally consistent adversaries]\label{def:Semi-ConssitentParties}
Let $\Pi=(\Pc_1,\ldots,\Pc_n)$ be an $n$-party protocol and let $\sset{\nxtmsg_{i,i'}^j}_{i,i' \in [n],j\in \N}$ be its set of next-message functions, \ie
\[
m^j_{i,i'}=\nxtmsg_{i,i'}^j\left(b;r;(m^1_{1,i},\ldots,m^1_{n,i}),\ldots, (m^{j-1}_{1,i},\ldots,m^{j-1}_{n,i})\right)
\]
is the message party $\Pc_i$ sends to party $\Pc_{i'}$ in the \jth round, given that its input bit is $b$, the random coins it flipped till now are $r$, and in round $j' < j$, it got the message $m^{j'}_{i'',i}$ from party $\Pc_{i''}$. An adversary taking the role of $\Pc_i$ is said to be {\sf locally consistent} \wrt $\Pi$, if it \emph{flips its random coins honestly}, and the message it sends in the \jth round to party $\Pc_{i'}$ takes one of the following two forms:

\begin{description}
\item[Abort:] the message $\perp$.

\item[Input and message selection:] a set of messages $\set{m_\ell}_{\ell=1}^k$, for some $k$, such that for each $\ell \in [k]$:
\[
m_\ell = \nxtmsg_{i,i'}^j\left(b_\ell;r;((m^1_{1})_\ell,\ldots,(m^1_{n})_\ell),\ldots, ((m^{j-1}_{1})_\ell,\ldots,(m^{j-1}_{n})_\ell)\right),
\]
where $b_\ell\in \zo$, $r$ are the coins $\Pc_i$ tossed (honestly) until now, and $(m^{j'}_{i''})_\ell$, for each $j'<j$ and $i''\neq i$, is one of the messages $\Pc_i$ received from party $\Pc_{i''}$ in the \iith{j} round (or the empty string).
\end{description}
\end{definition}

That is, a locally consistent party $\Pc_{i}$ might send party $\Pc_{i'}$ a sequence of messages
(and not just one as instructed), each consistent with a possible choice of its input bit, and some of the messages it received in the previous round. In turn, this will enable party $\Pc_{i'}$, if corrupted, the freedom to choose in the next rounds the message of $\Pc_{i}$ it would like to act according to. Note that \wlg, $\Pc_{i}$ will always send a single message to the honest parties, as otherwise they will discard the messages.

A few remarks are in place.
\begin{enumerate}
	\item While the above definition does not enforce between-rounds consistency (a party might send to another party a first-round message consistent with input $0$ and a second-round message consistent with input $1$), compiling a given protocol so that every message party $\Pc_{i}$ sends to $\Pc_{i'}$ contains the previous messages $\Pc_{i}$ sent to $\Pc_{i'}$, will enforce such between-rounds consistency on locally consistent parties.
	
	\item Although a locally consistent adversary tosses its random coins honestly, he may toss all random coins at the beginning of the protocol and choose its actions as a function of these coins. Our attacks in \cref{sec:FirstRound,sec:SecondRound} do not take advantage of this capability, and let the corrupted parties toss the random coins for a given round at the beginning of the round.

	\item Using standard cryptographic techniques, a protocol secure against locally consistent adversaries can be compiled into one secure against arbitrary malicious adversaries, without hurting the efficiency of the protocol ``too much,'' and in particular preserve the round complexity (see \cref{sec:intro:LocalToFull}).
	
	\item The locally consistent parties considered in \cref{sec:FirstRound,sec:SecondRound} do not take full advantage of the generality of \cref{def:Semi-ConssitentParties}. Rather, the parties considered either act honestly but abort at the conclusion of the first round, cheat in the first round and then abort, or cheat only in the second round and then abort.
\end{enumerate}

\subsubsection{Public-Randomness Protocols}\label{sec:OurResults:PR}
In \cref{sec:intro:model}, we showed that the description of many natural protocols can be simplified when security is required to hold only against locally consistent adversaries. In this relaxed description a trusted setup phase and cryptographic assumptions are not required, and every party can publish the coins it locally tossed in each round.

\begin{definition}[Public-randomness protocols]\label{def:PR}
A protocol has {\sf public randomness}, if every party's message consists of two parts: the randomness it sampled in that round, and an arbitrary message which is a function of its view (input, incoming messages, and coins tossed up to and including that point). The party's first message also contains its setup parameters, if such exist.
\end{definition}

\subsubsection{Byzantine Agreement}\label{sec:BA}

We now formally define the notion of Byzantine agreement. Since we focus on lower bounds we will consider only the case of a single input bit and a single output bit. A more general notion of Byzantine agreement will include string input and string outputs. A generic reduction shows that the cost of agreeing on strings rather than bits is two additional rounds~\cite{TC84}.

\begin{definition}[Byzantine Agreement]\label{def:BA}
We associate the following properties with a \ppt $n$-party Boolean input/output protocol $\Pi$.

\begin{description}
	\item[Agreement.] Protocol $\Pi$ has {\sf $(t,\alpha)$-\Agr}, if the following holds \wrt any \ppt adversary controlling at most $t$ parties in $\Pi$ and any value of the non-corrupted parties' input bits: in a random execution of $\Pi$ on sufficiently large security parameter, all non-corrupted parties output the \emph{same} bit with probability at least $1-\alpha$.\footnote{A more general definition would allow the parameter $\alpha$ (and the parameters $\beta,\gamma$ below) to depend on the protocol's security parameter. But in this paper we focus on the case that $\alpha$ is a fixed value.}
	
	\item[Validity.] Protocol $\Pi$ has {\sf $(t,\beta)$-\Vld}, if the following holds \wrt any \ppt adversary controlling at most $t$ parties in $\Pi$ and an input bit $b$ given as input to all non-corrupted parties: in a random execution of $\Pi$ on sufficiently large security parameter, all non-corrupted parties output $b$ with probability at least $1-\beta$.
	
	\item[Halting.] Protocol $\Pi$ has {\sf $(t,q,\gamma)$-\Halt}, if the following holds \wrt any \ppt adversary controlling at most $t$ parties in $\Pi$ and any value of the non-corrupted parties' input bits: in a random execution of $\Pi$ on sufficiently large security parameter, all non-corrupted parties halt within $q$ rounds with probability at least $\gamma$.
\end{description}
	
\noindent
Protocol $\Pi$ is a $(t,\alpha,\beta,q,\gamma)$-\BA, if it has $(t,\alpha)$-\Agr, $(t,\beta)$-\Vld, and $(t,q,\gamma)$-\Halt. If the protocol has a setup phase, then the above probabilities are taken \wrt this phase as well.
\end{definition}

\begin{remark}[Concrete security]
Since we care about fixed values of a protocol's characteristics (\ie agreement), the role of the security parameter in the above definition is to enable us to bound the running time of the parties and adversaries in consideration in a meaningful way, and to parametrize the cryptographic tools used by the parties (if there are any). Since the attacks we present are efficient assuming the protocol is efficient (in any reasonable sense), the bounds we present are applicable for a fixed protocol that might use a \emph{fixed} cryptographic primitive, \eg SHA-256.
\end{remark}

\subsection{The Bounds}\label{sec:OurResult:Bounds}
We proceed to present the formal statements of the three lower bounds.
Recall that Byzantine agreement cannot be achieved for $t\geq n/2$, since otherwise the corrupted parties can simply play honestly on an input of their choice and force the output. We therefore consider $t<n/2$ throughout the paper.

\paragraph{First-round halting, arbitrary protocols.}
The first result bounds the halting probability of arbitrary protocols after a single round. Namely, for ``small'' values of $\alpha$ and $\beta$, the halting probability is ``small'' for $t\geq n/3$ and ``close to $1/2$'' for $t\geq n/4$.

\def\ThmFirstRoundArb
{
Let $\Pi$ be a \ppt $n$-party protocol that is $(t,\alpha,\beta,1,\gamma)$-\BA against locally consistent, static, non-rushing adversaries. Then,
\begin{itemize}
	\item $t \ge n/3$ implies $\gamma \le 6\alpha + 2\beta +\err$
	\item $t \ge n/4$ implies $\gamma \le 1/2 + 5\alpha + \beta + \err$,
\end{itemize}
for $\err= \FirstErr$ ($\err=0$ for public-randomness protocols whose security holds against rushing adversaries).
}
\begin{theorem}[restating \cref{thm:intro:FirstRound}]\label{thm:FirstRound:Arb}
	\ThmFirstRoundArb
\end{theorem}

\paragraph{Second-round halting, arbitrary protocols.}
The second result bounds the halting probability of arbitrary protocols after two rounds.

\def\ThmSecondRoundArb
{
Let $\Pi$ be a \ppt $n$-party protocol that is $(t,\alpha,\beta,2,\gamma)$-\BA against locally consistent, static, non-rushing adversaries for $t > n/4$. Then $\gamma \le 1 + 2\alpha + \frac\beta{w^2} -\frac1{2w^2}$ for $w = \ceil{(n-\ceil{n/4})/\floor{t - n/4}}+1$.
}
\begin{theorem}[restating \cref{thm:intro:SecondRound:Arb}]\label{thm:SecondRound:Arb}
	\ThmSecondRoundArb
\end{theorem}

In particular, for $t = (1/4 + \eps)\cdot n$ and ``small'' $\alpha$ and $\beta$, the protocol might not halt at the conclusion of the second round with probability $\approx \eps^2$.

\paragraph{Second-round halting, public-randomness protocols.}
The third result bounds the halting probability of public-randomness protocols after two rounds. The result requires adaptive and rushing adversaries, and is based on \cref{con:IsoBot} (stated in \cref{sec:Iso} below).

\newcommand{\ept}{\eps_t}
\newcommand{\epg}{\eps_\gamma}

\def\ThmSeconRoundPR
{
Assume \cref{con:IsoBot} holds, then for any (constants) $\ept,\epg>0$ there exists $\alpha> 0$ such that the following holds for large enough $n$: let $\Pi$ be a \ppt $n$-party, public-randomness protocol that is $(t,\alpha,\beta= \epg^2/200,2,\gamma)$-\BA against locally consistent, rushing, adaptive adversaries. Then,
\begin{itemize}
	\item $t \ge (1/3 + \ept)\cdot n$ implies $\gamma < \epg$.
	\item $t \ge (1/4 + \ept)\cdot n$ implies $\gamma < \frac12 + \epg$.
\end{itemize}
}

\begin{theorem}[restating \cref{thm:intro:SecondRound:PR}] \label{thm:SecondRound:PR}
\ThmSeconRoundPR
\end{theorem}
In particular, assuming the protocol has perfect agreement and validity, the protocol never halts in two rounds if the fraction of corrupted parties is greater than $1/3$, and halts in two rounds with probability at most $1/2$ if the fraction of corrupted parties is greater than $1/4$.

The value of $\alpha$ in the theorem is (roughly) $\delta\cdot \ept\cdot \epg^2$ where $\delta$ is the constant guaranteed by \cref{con:IsoBot}. We were not trying to optimize over the constants in the above statement, and in particular it seems that $\beta$ can be pushed to $\epg^2$.

\subsection{The Combinatorial Conjecture}\label{sec:Iso}
Next, we provide the formal statement for the combinatorial conjecture used in \cref{thm:SecondRound:PR}.
For $n\in \N$ and $\sigma \in [0,1]$, let $\bns$ be the distribution induced on the subsets of $[n]$ by sampling each element independently with probability $\sigma$.
For a finite alphabet $\Sigma$, a vector $\vx\in \sn$, and a subset $\cS \subseteq [n]$, define the vector $\bot_\cS(\vx) \in \sn$ by
\[
\bot_\cS(\vx)_i=
\begin{cases}
\bot, & i\in \cs,\\
\vx_i, & \text{otherwise}.
\end{cases}
\]

\begin{conjecture}[restating \cref{con:intro:IsoBot}]\label{con:IsoBot}
For any $\sigma,\lambda >0$ there exists $\delta>0$ such that the following holds for large enough $n\in \N$. Let  $\Sigma$ be a finite alphabet and let $\cA_0,\cA_1 \subseteq \sbn$  be  two  sets such that for both $b\in \zo$:
\begin{align*}
\ppr{\cs\gets \bns}{\ppr{\vr \gets \Sigma^n}{\vr,\bot_{\cS}(\vr) \in \cA_b} \ge  \lambda } \ge 1-\delta.
\end{align*}
Then,
\begin{align*}
\ppr{\substack{\vr\gets \Sigma^n\\ \cS \gets \bns}}{\forall b\in \zo\colon  \set{\vr,\bot_{\cS}(\vr)}  \cap \cA_b \neq \emptyset} \ge  \delta.
\end{align*}
\end{conjecture}

\section{Lower Bounds on First-Round Halting}\label{sec:FirstRound}

In this section, we present our lower bound for the probability of first-round halting in Byzantine agreement protocols.
\begin{theorem}[Bound on first-round halting. \cref{thm:FirstRound:Arb} restated]\label{thm:FirstRound:Arb:Res}
\ThmFirstRoundArb
\end{theorem}

Let $\Pi$ be as in \cref{thm:FirstRound:Arb:Res}.
\Wlg and for ease of notation, we denote by $\Pi$ the modified protocol that outputs $\bot$ if a party does not halt after the first round (it will be clear that the attack, described below, does not benefit from this change).
We also omit the security parameter from the parties' input list, it will be clear though that the adversaries we present are efficient \wrt the security parameter.

\begin{lemma}[Neighboring executions]\label{lemma:FirstRound:Arb}
Let $\vv,\vv' \in \zn$ be with $\ham(\vv,\vv') \le t$. Then for both $\o\in \zo$:
\[
\pr{\Pi(\vv') \in \omin^n \setminus \mon} \ge	\pr{\Pi(\vv) \in \omin^n} - (1- \gamma) - 4\alpha - \err.
\]
\end{lemma}
Namely, the lemma bounds from below the probability that in a random honest execution of the protocol on input $\vv'$, at least one party halts in the first round while outputting $\o$.

We prove \cref{lemma:FirstRound:Arb} below, but first use it to prove \cref{thm:FirstRound:Arb:Res}. We also make use of the following immediate observation.
\begin{claim}[Almost pre-agreement]\label{claim:FirstRoundound:Arb:Validity}
Let $\vv \in \zn$ and $\o \in \zo$ be such that $\ham(\vv,\o^n) \le t$. Then, $\pr{\Pi(\vv) \in \omin^n} \ge 1 - \alpha - \beta$.
\end{claim}
\begin{proof}
Let $\cA \subset [n]$ be a subset of size $n-t$ such that $\vv_\cA = \o^{\size{\cA}}$. The claimed validity of $\Pi$ yields that
\[
\pr{\Pi(\vv)_\cA \notin \omin^{\size{\cA}}} \leq \beta.
\]
This follows from $\beta$-validity of $\Pi$ and the fact that an honest party cannot distinguish between an execution of $\Pi(\vv)$ and an execution of $\Pi(b^n)$ in which all parties not in $\cA$ act as if their input bit is as in $\vv$. Hence, by the claimed agreement of $\Pi$,
\[
\pr{\Pi(\vv) \notin \omin^n} \leq \alpha + \beta.
\]
\end{proof}

\begin{proof}[Proof of \cref{thm:FirstRound:Arb:Res}]~
We separately prove the theorem for $t \ge n/3$ and for $t \ge n/4$.

\paragraph{The case $t \ge n/3$.}
Let $\vz= 0^t 1^{\cnt} 0^{\fnt} $ and $\vo = 1^t 1^{\cnt} 0^{\fnt}$. Note that $\ham(\vz,\vo) = t$, and that for both $\o\in \zo$ it holds that $\ham(\vv_\o,\o^n)\le t$. Hence, by \cref{claim:FirstRoundound:Arb:Validity}, for both $\o\in \zo$:
\begin{align*}
\pr{\Pi(\vv_\o) \in \omin^n} \ge 1 - \alpha - \beta.
\end{align*}
Applying \cref{lemma:FirstRound:Arb} to $\vv= \vz$ and $\vv'= \vo$ yields that
\begin{align*}
\pr{\Pi(\vo) \in \set{0,\perp}^n \setminus \mon}
& \geq \pr{\Pi(\vz) \in \set{0,\perp}^n} - (1- \gamma) - 4\alpha - \err \\
& \ge 1 - 5\alpha - \beta - (1- \gamma) - \err.
\end{align*}
Since by \cref{claim:FirstRoundound:Arb:Validity} it holds that $\pr{\Pi(\vo) \in \set{0,\perp}^n \setminus \mon} \leq \pr{\Pi(\vo) \notin \set{1,\perp}^n} \leq \alpha+\beta$, we conclude that
$6\alpha + 2\beta + (1-\gamma) + \err\ge 1$, hence $\gamma \leq 6\alpha + 2\beta + \err$.

\paragraph{The case $t \ge n/4$.} In this case there are no two vectors that are $t$ apart in Hamming distance, and still each of them has $n-t$ entries of opposite values. Rather, we consider the two vectors $\vz= 0^t 0^t 0^t 1^{n-3t}$ and $\vo= 1^t 1^t 0^t 1^{n-3t}$ of distance $2t$. For both $\o\in \zo$, the vector $\vv_\o$ has at least $n-t$ entries with $\o$ and is of distance $t$ from the vector $\vstar =1^t 0^t 0^t 1^{n-3t}$.

As in the first part of the proof, Applying \cref{claim:FirstRoundound:Arb:Validity,lemma:FirstRound:Arb} on $\vv_b$ and $\vstar$, for both $b\in \zo$, yields that
\begin{align*}
\pr{\Pi(\vstar) \in \omin^n \setminus \mon}
& \geq \pr{\Pi(\vv_b) \in \set{b,\perp}^n} - (1- \gamma) - 4\alpha - \err \\
& \ge 1 - 5\alpha - \beta - (1- \gamma) - \err.
\end{align*}
By union bound, we conclude that $2(5\alpha + \beta + (1- \gamma) + \err) \ge 1$, hence $\gamma \le 1/2+5\alpha + \beta+\err$.
\end{proof}

\newcommand{\PPf}{\Pi^\cP}	
\subsection{Proving Lemma~\ref{lemma:FirstRound:Arb}}
\begin{proof}[Proof of \cref{lemma:FirstRound:Arb}]
Fix $b \in \zo$ and let $\delta = \pr{\Pi(\vv) \in \omin^n}$. Let $\cP$ be the coordinates in which $\vv$ and $\vv'$ differ, and let $\oP = [n] \setminus \cP$. Let $I$ be the index (a function of the parties' coins and setup parameters) of the smallest party in $\oP$ that halts in the first round and outputs \emph{the same value}, both if the parties in $\cP$ send their messages according to input $\vv$ and if they do that according to $\vv'$. We let $I=0$ if there is no such party, and (abusing notation) sometimes identify $I$ with the event that $I\neq 0$, \eg $\pr{I}$ stands for $\pr{I\ne 0}$. By definition,
\begin{align*}
\delta \le \pr{\Pi(\vv) \in \omin^n \qand I}+ (1- \pr{I})
\end{align*}
and thus
\begin{align}\label{eq:FirstRound:Arb:1}
\pr{\Pi(\vv) \in \omin^n \qand I} \ge \delta - (1-\pr{I})
\end{align}
	
It follows that
\begin{align}\label{eq:FirstRound:Arb:2}
\pr{\Pi(\vv') \in \omin^n \setminus \mon} &\ge \pr{\Pi(\vv') \in \omin^n \qand I}\\
&=\pr{\Pi(\vv') \in \omin^n \qand \Pi(\vv')_I= b}\nonumber\\
&\ge \pr{\Pi(\vv')_I= b} - \alpha\nonumber\\
&= \pr{\Pi(\vv)_I = \o} - \alpha\nonumber\\
&\ge \pr{\Pi(\vv) \in \omin^n \qand \Pi(\vv)_I = \o} - 2\alpha\nonumber\\
&= \pr{\Pi(\vv) \in \omin^n \qand I} - 2\alpha\nonumber\\
&\ge \delta - (1-\pr{I}) - 2\alpha.\nonumber
\end{align}
The first inequality and the equalities hold by the definition of $I$. The second and third inequalities hold by agreement, and the last inequality holds by \cref{eq:FirstRound:Arb:1}. We conclude the proof showing that:
\begin{align}\label{claim:FirstRound:Arb}
\pr{I} \ge \gamma - \err - 2\alpha
\end{align}

\newcommand{\ea}{F}
\newcommand{\eh}{C}

Let $\eh$ denote the event (a function of the parties' coins and setup parameters) that for each party $j$ in $\oP$ there exists an input in $\set{\vv,\vv'}$ on which it does not halt.
Furthermore, let $\ea \assign \neg \eh \land (I=0)$, \ie there exists a party that halts on both inputs but outputs different values. By definition, $I=0$ is equivalent to the event $\ea \lor \eh$.

Consider the adversary that in the first round acts toward a random subset  $\oP' \subseteq \oP$ according to input $\vv$, towards the remaining parties according to $\vv'$, and aborts at the end of this round. Fix some random coins and setup parameters in $\ea$, and let $i\in \oP$ be a party that, under this fixing, halts in the first round on both  $\vv$ and $\vv'$, but outputs a different value. Note that  the other parties in $\oP$ cannot distinguish whether $i$ is in $\oP'$ or not (in both cases $i$ halts at the end of  the first round). Since, by assumption,  $t< n/2$ (\ie there exist additional honest parties), it follows that under the above  conditioning, agreement is violated with probability at least $1/2$.  We conclude that $\pr{\ea} \le 2\alpha$.

It is also clear that when $\eh$ occurs, the above attacker fails to prevent an honest party in $\oP$ from halting in the first round only if the following event happens: each party in $\oP $ does not halt in $\Pi(\vv'')$ for some $\vv'' \in \set{\vv,\vv'}$, but the adversary acts towards each of these parties on the input in which it does halt. The latter event happens with probability at most $2^{-\size{\oP}} \le \FirstErr = \err$. Thus, $\pr{\eh} \le 1 - (\gamma - \err)$. We conclude that
\begin{align}
\pr{I} \ge 1 - \pr{\eh} - \pr{\ea} \ge \gamma - \err - 2\alpha
\end{align}

Finally, we note that if the protocol has public randomness, the (now rushing) attacker does not have to guess what input to act upon.
Rather, after seeing the first-round randomness, it \emph{finds} an input $\vv'' \in\set{\vv,\vv'}$ such that at least one party in $\oP$ does not halt in $\Pi(\vv'')$ or violates agreement, and acts according to this input. Specifically, given the honest parties' first-round coins, the attacker can compute on its own all honest-to-honest first-round messages (recall that we consider private channels, so the attacker does not see those messages on the channels), and locally check which honest party will halt with output 0 and which will halt with output 1 when playing according to $\vv$ and when playing according to $\vv'$. Hence, the bound on $I$ changes to
\begin{align*}
\pr{I} \ge \gamma - \alpha,
\end{align*}
proving the theorem statement for such protocols.
\end{proof}

\newcommand{\DF}{D_\F}
\newcommand{\SDF}{\Supp(D_\F)}

\section{Lower Bounds on Second-Round Halting}\label{sec:SecondRound}
In this section, we prove lower bounds for second-round halting of Byzantine agreement protocols. In \cref{sec:TwoRoundProtcol:Arbitrary}, we prove a bound for arbitrary protocols, and in \cref{sec:TwoRoundProtcol:PR}, we give a much stronger bound for public-randomness protocols (the natural extension of public-coin protocols to the ``with-input'' setting).

\subsection{Arbitrary Protocols}\label{sec:TwoRoundProtcol:Arbitrary}
We start by proving our lower bound for second-round halting of arbitrary protocols.

\begin{theorem}[Bound on second-round halting, arbitrary protocols. \cref{thm:SecondRound:Arb} restated]\label{thm:SecondRound:Arb:Res}
\ThmSecondRoundArb
\end{theorem}

Let $\Pi$ be as in \cref{thm:SecondRound:Arb:Res}.
\Wlg and for ease of notation, we denote by $\Pi$ the modified protocol that outputs $\bot$ if a party does not halt after the first two rounds (it will be clear that the attack, described below, does not benefit from this change). We also assume \wlg that the honest parties in an execution of $\Pi$ never halt in the first round (by adding a dummy round if needed). Finally, we omit the security parameter from the parties' input list, it will be clear though that the adversaries we present are efficient \wrt the security parameter.

Let $k =\ceil{n/4}$ and let $h = \ceil{(n-k)/(t-k)}$. The theorem is easily implied by the next lemma.

\begin{lemma}[Neighboring executions]\label{lemma:SecondRound:Arb}
Let $\vv,\vv' \in \zn$ be with $\ham(\vv,\vv') \le k$. Then, for every $\o\in \zo$:
\[
\pr{\Pi(\vv') = \o^n} \ge \pr{\Pi(\vv) = \o^n} - h(h+1)(2\alpha + 1-\gamma) - \alpha.
\]
\end{lemma}
Namely, the lemma bounds from below the probability that in a random honest execution of the protocol on input $\vv'$ all parties halt within two rounds while outputting $\o$.

We prove \cref{lemma:SecondRound:Arb} below, but first use it to prove \cref{thm:SecondRound:Arb:Res}. We also make use of the following immediate observation.
\begin{claim}[Almost pre-agreement]\label{claim:SecondRound:Arb:Validity}
Let $\vv \in \zn$ and $\o \in \zo$ be such that $\ham(\vv,\o^n) \le t$. Then, $\pr{\Pi(\vv) = \o^n } \ge 1 - \alpha - \beta - (1- \gamma)$.
\end{claim}
\begin{proof}
The same argument as in the proof of \cref{claim:FirstRoundound:Arb:Validity} yields that
\[
\pr{\Pi(\vv) \notin \omin^n} \leq \alpha + \beta.
\]
Thus, by $\gamma$-second-round halting
\[
\pr{\Pi(\vv) \ne \o^n} \leq \alpha + \beta + (1-\gamma).
\]
\end{proof}

\begin{proof}[Proof of \cref{thm:SecondRound:Arb:Res}]
Consider the vectors $\vz= 0^k 0^k 0^k 1^{n-3k}$, $\vo= 1^k 1^k 0^k 1^{n-3k}$ and $\vstar =1^k 0^k 0^k 1^{n-3k}$. Note that for both $\o\in \zo$ it holds that $\ham(\vv_b,\o^n)\le t$ since $n/4 \leq k\leq t$), and that $\ham(\vv_b,\vstar)=k$. Applying \cref{lemma:SecondRound:Arb,claim:SecondRound:Arb:Validity} for each of these vectors, yields that for both $\o\in \zo$:
\begin{align*}
\pr{\Pi(\vstar) = \o^n} &\ge 1- \alpha - \beta -(1- \gamma) - h(h+1)(2\alpha + 1-\gamma) -\alpha\\
&\ge 1- \beta - (h+1)^2(2\alpha + 1-\gamma).
\end{align*}
Note that $w=h+1$, which implies $\beta +w^2(2\alpha + 1-\gamma) \ge 1/2$, and the proof follows by a simple calculation.
\end{proof}

\subsubsection{Proving \cref{lemma:SecondRound:Arb}}
We assume for ease of notation that $\ham(\vv,\vv')=k$ (rather than $\le k$) and let $\ell=t-k$. Assume for ease of notation that $h\cdot \ell = n-k $ (\ie no rounding), and for a $k$-size subset of parties $\cP \subset [n]$, let $ \cL^\cP_1,\dots,\cL_h^\cP$ be an arbitrary partition of $\oP = [n] \setminus \cP$ into $\ell$-size subsets. Consider the following family of protocols:

\newcommand{\PP}{\Pi^\cP}

{ \samepage
\begin{protocol}[$\PP_{d}$]\label{prot:main:Arb} ~
\begin{description}
	\item [Parameters:] A subset $\cP\subseteq [n]$ and an index $d\in (h)$.
	\item [Input:] Every party $\Pc_i$ has an input bit $v_i\in\zo$.
	
	\item [First round:] ~
	\begin{description}
		\item[Party $\Pc_i \in \cP$.]
		If $d=0$ [\resp $d=h$], act honestly according to $\Pi$ \wrt input bit~$v_i$ [\resp $1-v_i$].
		Otherwise,		
		\begin{enumerate}
			\item Choose random coins honestly (\ie uniformly at random).
			
			\item To each party in $\bigcup_{j \in \set{1,\ldots,d}} \cL_j^\cP$: send a message according to input $1-v_i$.
			
			\item To each party in $\bigcup_{j \in \set{d+1 ,\ldots, h}} \cL_j^\cP$: send a message according to input $v_i$ (real input).
			
			\item Send no messages to the other parties in $\cP$.
		\end{enumerate}
		
		\item[Other parties.] Act according to $\Pi$.
	\end{description}
	
	\item [Second round:]~
	
	\begin{description}
		\item[Party $\Pc_i \in \cP$.]
        If $d=0$ [\resp $d=h$], act honestly according to $\Pi$ \wrt input bit~$v_i$ [\resp $1-v_i$]; otherwise, abort.
		
		\item[Other parties.] Act honestly according to $\Pi$.
	\end{description}
\end{description}
\end{protocol}
}

Namely, the ``pivot'' parties in $\cP$ gradually shift their inputs from their real input to its negation according to parameter $d$. Note that protocol $\PP_{0}(\vv)$ is equivalent to an honest execution of protocol $\Pi(\vv)$, and $\PP_{h}(\vv)$ is equivalent to an honest execution of $\Pi(\vv')$, for $\vv'$ being $\vv$ with the coordinates in $\cP$ negated. Note that for ``intermediately'' protocols $\PP_d$ for $0<d<h$, the pivot parties send conflicting messages to honest parties in the first round and abort in the second round. The reason that aborting in the second round does not affect our analysis below is that, without loss of generality, honest parties can exchange their views in the second round, realize the pivot parties are cheating (as we consider locally consistent adversaries), and ignore their messages.
\cref{lemma:SecondRound:Arb} easily follows by the next claim about \cref{prot:main:Arb}. In the following we let $\delta_b = \pr{\Pi(\vv)_\oP = \o^{\size{\oP}}}$.

\newcommand{\ds}{{d^\ast}}
\begin{claim}\label{claim:SecondRound:Arb}
For every $k$-size subset $\cP\subset [n]$, $b\in\zo$ and $d\in (h)$, it holds that
\[
\pr{\PP_d(\vv)_\oP = \o^{\size{\oP}}} \ge \delta_b - d (h+1)(2\alpha + 1-\gamma).
\]
\end{claim}
We prove \cref{claim:SecondRound:Arb} below, but first use it to prove \cref{lemma:SecondRound:Arb}.
\begin{proof}[Proof of \cref{lemma:SecondRound:Arb}]
By \cref{claim:SecondRound:Arb},
\begin{align*}
\pr{\PP_h(\vv)_\oP = \o^{\size{\oP}}} \ge \delta_b - h(h+1)(2\alpha + 1-\gamma).
\end{align*}
Recall that $\PP_h(\vv)_\oP = \o^{\size{\oP}}$ only when parties complete the protocol in the second round, since, by assumption, a party that continues to beyond the second round outputs $\bot$. In addition, since $\PP_h(\vv)$ is just an honest execution of $\Pi(\vv')$, by agreement it holds that
\[
\pr{\Pi(\vv') = \o^n} \ge \delta_b - h(h+1)(2\alpha + 1-\gamma) - \alpha.
\]
\end{proof}

\begin{proof}[Proof of \cref{claim:SecondRound:Arb}]
The proof is by induction on $d$. The base case $d=0$ holds by definition. Suppose for contradiction the claim does not hold, and let $\ds\in(h-1)$ be such that the claim holds for $\ds$ but not for $\ds+1$. Let $\gamma_d$ be the probability that all honest parties halt in the second round of a random execution of $\PP_d(\vv)$. Since the claim holds for $\ds$, it holds that
\begin{align}\label{eq:main:Arb:01}
\pr{\PP_\ds(\vv)_\oP = \o^{\size{\oP}}} \ge \delta_b - \beta- \ds (h+1)(2\alpha + 1-\gamma)
\end{align}
Since the claim does not hold for $\ds+1$, but  all honest parties output something in $\PP_{\ds+1}$ with probability at least $\gamma_{\ds+1}$, we have that
\begin{align}\label{eq:main:Arb:02}
\pr{\PP_{\ds+1}(\vv)_\oP \in \zo^{\size{\oP}} \setminus \sset{\o^{\size{\oP}}}} > 1-\left(\delta_b - \beta - (\ds+1) (h+1)(2\alpha + 1-\gamma)\right)- (1-\gamma_{\ds+1})
\end{align}

We note that for every $d\in (h)$
\begin{align}\label{eq:main:Arb:03}
\frac{1- \gamma_d}{h+1} \le 1-\gamma
\end{align}

Indeed, otherwise, the adversary that corrupts the parties in $\cP$ and acts like $\PP_d$ for a random $d\in (h)$, violates the $\gamma$-second-round-halting property of $\Pi$. We conclude that
\begin{align}\label{eq:main:Arb:1}
\lefteqn{\ppr{\vr}{\PP_\ds(\vv;\vr)_\oP = \o^{\size{\oP}} \qand \PP_{\ds+1}(\vv;\vr)_\oP \in \zo^{\size{\oP}} \setminus \sset{\o^{\size{\oP}}}}}\\
& \ge 1 - \left(1- \ppr{\vr}{\PP_\ds(\vv;\vr)_\oP = \o^{\size{\oP}}} \right) - \left(1- \ppr{\vr}{\PP_{\ds+1}(\vv;\vr)_\oP \in (\zo^{\size{\oP}} \setminus \sset{\o^{\size{\oP}}}}\right)\nonumber\\
& > (h+1) (2\alpha + 1 - \gamma) - (1-\gamma_{\ds+1}) \nonumber\\
& \ge 2\alpha(h+1),\nonumber
\end{align}
for $\vr$ being the randomness of the parties. The second inequality is by \cref{eq:main:Arb:01,eq:main:Arb:02}, and the third one by \cref{eq:main:Arb:03}.

Consider the adversary \Ac that samples $d\gets (h-1)$, corrupts the parties in $\cP \cup \cL^\cP_{d+1}$, and acts towards a uniform random subset of the honest parties according to $\PP_d$ and to the remaining parties according to $\PP_{d+1}$. Since \Ac violates agreement if it guesses $d= \ds$ and it partitions the honest parties suitably, \cref{eq:main:Arb:1} yields that \Ac causes disagreement with probability larger than $2\alpha(h+1)/(2(h+1)) = \alpha$.  Since \Ac corrupts $\ssize{\cP \cup \cL^\cP_{d+1}}\leq t$ parties, this contradicts the assumption about $\Pi$.
\end{proof}

\subsection{Public-Randomness Protocols}\label{sec:TwoRoundProtcol:PR}
We proceed to prove our lower bound for second-round halting of public-randomness protocols.

\begin{theorem}[Lower bound on second-round halting, public-randomness protocols. \cref{thm:SecondRound:PR} restated] \label{thm:SecondRound:PR:Res}
\ThmSeconRoundPR
\end{theorem}

Assume \cref{con:IsoBot} holds. Let $\Pi$ be as in the theorem statement, and assume $\gamma= \epg$ in the case $t \ge (1/3 +\ept)\cdot n$ and $\gamma= \half + \epg$ in the case $t \ge (1/4 +\ept)\cdot n$. Let $\lambda =\epg /10$ and $\sigma= \ept/4$. Recall that $\bot_\cS(\vx)$ is the string resulting by replacing all entries of $\vx$ indexed by $\cS$ with $\bot$.
\cref{con:IsoBot} yields that there exists $\delta>0$ such that the following holds for large enough $n$: let $\Sigma$ be a finite alphabet and let $\cA_0,\cA_1 \subset \sbn$ be two sets such that for both $b\in \zo$:
\begin{align*}
\ppr{\cs\gets \bns}{\ppr{\vr \gets \sn}{\vr,\bot_{\cS}(\vr) \in \cA_b} \ge \lambda} \ge 1-\delta.
\end{align*}

Then,
\begin{align}\label{eq:IsoBot}
\ppr{\vr\gets \sn,\cS \gets \bns}{\forall b\in \zo\colon \set{\vr,\bot_{\cS}(\vr)} \cap \cA_b \neq \emptyset} \ge \delta.
\end{align}

In the following we assume $\alpha = \min\set{\delta\lambda \ept/10, \beta}$ and derive a contradiction, yielding that the agreement error has to be larger than that.

Fix $n$ that is large enough for \cref{eq:IsoBot} to hold and that (by Chernoff bound) $\ppr{\cs \gets \bns}{\size{\cs} > 2\sigma n} = 2^{- \Theta(n\cdot \sigma)} \le \alpha$, \ie $n> \Theta((\log 1/\alpha)/\sigma)$. As in the proof of \cref{thm:SecondRound:Arb:Res}, we assume for ease of notation that an honest party that runs more than two round outputs $\perp$, and that the honest parties in $\Pi$ never halt in one round. We also omit the security parameter from the parties input list. We assume \wlg that in the first round, the parties flip no coin, since such coins can be added to the setup parameter.

We use the following notation: the setup parameter and second-round randomness of the parties in $\Pi$ are identified with elements of $\cF$ and $\cR$, respectively. We denote by $f_i$ and $r_i$ the setup parameter and the second-round randomness of party $\Pc_i$ in $\Pi$, and let $\DF$ be the joint distribution of the parties' setup parameters (by definition, the joint distribution of the second-round randomness is the product distribution $\rn$). For $\vv \in \zn$, $\vf=(f_1,\ldots,f_n) \in \SDF$, and $\vr=(r_1,\ldots,r_n) \in\rn$, let $\Pi(\vv;(\vf,\vr))$ denote the execution of $\Pi$ in which party $\Pc_i$ gets input $v_i$, setup parameter $f_i$ and second-round randomness $r_i$. We naturally apply this notation for the variants of $\Pi$ considered in the proof.

For $\cs\subseteq[n]$, let $\Pi^\cs$ be the variant of $\Pi$ in which the parties in $\cs$ halt at the end of the first round. Let $k = \ceil{{t-\ept \cdot n}}$ (\ie $k = \ceil{n/3}$ if $t\ge (1/3 + \ept)\cdot n$, and $k = \ceil{n/4}$ if $t\ge (1/4 + \ept)\cdot n$). The heart of the proof lies in the following lemma.

\begin{lemma}[Neighboring executions]\label{lemma:SecondRound:PR}
Let $\vv,\vv' \in \zn$ be with $\ham(\vv,\vv') \le k$, let $\o\in \zo$, and let $\oS = [n] \setminus \cs$.  Then, with probability at least $\gamma- 7\lambda- \frac{\alpha + \pr{\Pi(\vv) \ne \o^n}}{\lambda}$ over $\vf \gets \DF$, it holds that
\begin{align*}
\ppr{\cs \gets \bns}{\ppr{\vr \gets \rn}{\Pi(\vv';(\vf,\vr)) = \o^n \qand \Pi^\cs(\vv';(\vf,\vr))_\oS = \o^{\size{\oS}} } \ge \lambda}\ge 1-\delta.
\end{align*}
\end{lemma}
Namely, in an execution of $\Pi(\vv')$, all honest parties halt after two rounds and output $\o$, regardless of whether a random subset of parties aborts after the first round. \cref{lemma:SecondRound:PR} is proven in \cref{sec:lemma:SecondRound:PR}, but let us first use it to prove \cref{thm:SecondRound:PR:Res}. We make use of the following immediate observation:
\begin{claim}[Almost pre-agreement]\label{claim:SecondRound:PR:Validity}
Let $\vv \in \zn$ and $\o \in \zo$ be such that $\ham(\vv,\o^n) \le t$. Then, $\pr{\Pi(\vv) \in \set{\o,\perp}^n } \ge 1 - \alpha - \beta $.
\end{claim}
\begin{proof}
The proof of this claim uses an identical argument as in the proof of \cref{claim:FirstRoundound:Arb:Validity}.	
\end{proof}

\paragraph{Proving \cref{thm:SecondRound:PR:Res}.}
\begin{proof}[Proof of \cref{thm:SecondRound:PR:Res}]~
We separately prove the case $t \ge (1/3 + \ept)\cdot n$ and $t \ge (1/4 + \ept)\cdot n$.

\paragraph{The case $t \ge (1/3 + \ept)\cdot n$.}
Let $\vz= 0^k 1^{\cnk} 0^{\fnk} $ and let $\vo= 1^k 1^{\cnk} 0^{\fnk} $. Note that $\ham(\vz,\vo)=k$ and that for both $\o\in \zo$ it holds that $\ham(\vv_b,\o^n) \le t$. We will use \cref{lemma:SecondRound:PR,claim:SecondRound:PR:Validity} to prove that $\Pi(\vv_1) = 0^n$ with noticeable probability, contradicting the validity of the protocol.

Recall that, in this case, $\gamma = \epg$, that $\lambda = \epg/10$ and $\alpha,\beta \le \epg^2/200 = \lambda^2/2$.  \cref{claim:SecondRound:PR:Validity} yields that for both $\o\in\zo$:
\begin{align}\label{eq:SecondRound:PR:1}
\pr{\Pi(\vv_\o) \ne \oo^n }\ge \pr{\Pi(\vv_\o) \in \set{\o,\perp}^n} \ge 1- \alpha - \beta \ge 1- \lambda^2
\end{align}
Applying \cref{lemma:SecondRound:PR} \wrt $\vz$ and $\vo$ and $b=0$, yields that with probability at least
\[
\gamma- 7\lambda- \frac{\alpha + \pr{\Pi(\vz) \ne 0^n}}{\lambda} \geq 3\lambda- \lambda=2 \lambda
\]
over $\vf \gets \DF$, it holds that (by discarding the probability over $\cS$ since the item below does not depend on $\cS$)
\begin{align*}
\ppr{\vr}{\Pi(\vo;(\vf,\vr)) = 0^n}\ge \lambda.
\end{align*}
Therefore, overall
\begin{align*}
\pr{\Pi(\vo) = 0^n}\ge 2 \lambda^2,
\end{align*}
in contradiction to \cref{eq:SecondRound:PR:1}.

\paragraph{The case $t \ge (1/4 + \ept)\cdot n$.} Consider the vectors $\vz= 0^k 0^k 0^k 1^{n-3k}$, $\vo= 1^k 1^k 0^k 1^{n-3k}$ and $\vstar =1^k 0^k 0^k 1^{n-3k}$. Note that for both $\o\in \zo$ it holds that $\ham(\vv_\o,\o^n) \le t$ and that $\ham(\vv_\o,\vstar)=k$. Applying \cref{lemma:SecondRound:PR,claim:SecondRound:PR:Validity} on $\vv_b$ and $\vstar$, for both $\o\in \zo$, yields that $\Pi^\cs(\vstar) = \o^n$ with noticeable probability over the choice of $\cs$. This will allow us to use \cref{con:IsoBot} to lowerbound the protocol's agreement.

Recall that the distribution $\bns$, from which set $\cs$ is sampled, is the distribution induced on the subsets of $[n]$ by sampling each element independently with probability $\sigma$. In addition, recall that in the case at hand ($t \ge (1/4 + \ept)\cdot n$), we assume that $\gamma=1/2+\eps_\gamma$.
A similar calculation to the previous case yields that by \cref{lemma:SecondRound:PR,claim:SecondRound:PR:Validity}, for both $\o\in \zo$: with probability at least $\frac12 + 2\lambda$ over $\vf \gets \DF $ it holds that
\begin{align*}
\ppr{\cs \gets \bns}{\ppr{\vr \gets \rn}{\Pi(\vstar;(\vf,\vr)) = \o^n \qand \Pi^\cs(\vstar;(\vf,\vr))_\oS = \o^{\size{\oS}} } \ge \lambda}\ge 1-\delta.
\end{align*}
It follows that there exists a set $\cT \subseteq \SDF$ with $\ppr{\vf \gets \DF}{\cT}\ge 4\lambda$, such that for every $\vf\in \cT$, for \emph{both} $\o\in \zo$:
\begin{align}\label{eq:mainPR:4}
\ppr{\cs \gets \bns}{\ppr{\vr \gets \rn}{\Pi(\vstar;(\vf,\vr)) = \o^n \qand \Pi^\cs(\vstar;(\vf,\vr))_\oS = \o^{\size{\oS}} } \ge \lambda}\ge 1-\delta
\end{align}

We assume \wlg that if a party gets $\perp$ as its second-round random coins, it aborts after the first round. For $\vr \in (\cR\cup\sset{\bot})^n$ let $ \cE(\vr)$ be the indices in $\vr$ of the value $\bot$.
For $\vf \in \SDF$ and $\o \in \zo$, let
\begin{align}
\cA^\vf_\o = \set{\vr \in \rbot \colon \Pi(\vstar;(\vf,\vr))_{\overline{\cE(\vr)}} = \o^{\size{\overline{\cE(\vr)}}}}
\end{align}
By \cref{eq:mainPR:4}, for $\vf \in \cT$ and $\o\in \zo$, it holds that
\begin{align}
\ppr{\cs\gets \bns}{\ppr{\vr \gets \rn}{\vr,\bot_{\cS}(\vr) \in \cA^\vf_\o} \ge\lambda} \ge 1-\delta
\end{align}
Hence by \cref{con:IsoBot}, see \cref{eq:IsoBot}, for $\vf \in \cT$ it holds that
\begin{align*}
\ppr{\vr\gets \rn,\cS \gets \bns}{\forall \o \in \zo\colon \set{\vr,\bot_{\cS}(\vr)}\cap \cA_b \neq \emptyset} > \delta.
\end{align*}
That is,
\begin{align}\label{eq:PR:3}
\ppr{\vr\gets \cR^n,\cS \gets \bns}{\forall \o\in \zo \quad \exists \cs_\o \in \set{\cs,\emptyset} \colon \Pi^{\cS_\o}(\vstar;(\vf,\vr))_{\overline{\cS_\o}}= \o^{ \size{\overline{\cS_\o}}}} >\delta
\end{align}

{\samepage
\noindent
Consider the following adversary:

\begin{algorithm}[$\Ac$]~
\begin{description}
	\item[Pre-interaction.]
	Corrupt a random subset $\cS \gets \bns $ conditioned on $\size{\cs} \le 2\sigma n$.
	
	\item[First round.] Act according to $\Pi$.
	\item[Second round.] Sample $\cs_0,\cs_1$ at random from $\set{\emptyset,\cs}$, and act towards some honest parties according to $\Pi^{\cs_0}$ and towards the others according to $\Pi^{\cs_1}$ .
\end{description}
\end{algorithm}
}
Recall that $n$ is chosen so that $\ppr{\cs \gets \bns}{\size{\cs} >2\sigma n} \leq \alpha$ and that $\alpha < \delta/2$.
By \cref{eq:PR:3}, the above adversary violates the agreement of $\Pi$ on input $\vstar$ with probability larger than $\delta - \ppr{\cs \gets \bns}{\size{\cs} >2\sigma n} \ge \delta - \alpha >\alpha$, in contradiction with the assumed agreement of $\Pi$.
\end{proof}

\newcommand{\VV}{\cV^{\cP}}
\renewcommand{\PP}{\Pi^{\cP,\cS}}
\newcommand{\ops}{{\overline{\cP \cup \cS}}}

\subsubsection{Proving \cref{lemma:SecondRound:PR}}\label{sec:lemma:SecondRound:PR}
Fix $\vv,\vv' \in \zn$ and $b\in \zo$ as in the lemma statement. We assume for simplicity that $\ham(\vv,\vv')=k$ (rather than $\le k$).
Let $\ell= \floor{(t -k)/2}$ and let $h = \ceil{(n-k)/\ell}$. Assume for ease of notation that $h\cdot \ell = n-k $ (\ie no rounding), and for a $k$-size subset of parties $\cP \subset [n]$, let $ \cL^\cP_1,\dots,\cL_h^\cP$ be an arbitrary partition of $\oP = [n] \setminus \cP$ into $\ell$-size subsets. Consider the following protocol family.

{ \samepage
\begin{protocol}[$\PP_d$]\label{prot:main:2} ~
\begin{description}
    \item [Parameters:] subsets $\cP,\cS \subseteq [n]$ and an index $d\in (h)$.
    \item [Input:] Party $\Pc_i$ has a setup parameter $f_i$ and an input bit $v_i$.

	\item [First round:] ~
	\begin{description}
		\item[Party $\Pc_i \in \cP$.] ~
		If $d=0$ [\resp $d=h$], act honestly according to  $\Pi$ \wrt input bit $v_i$ [\resp $1-v_i$].
		Otherwise,
		\begin{enumerate}
			\item Choose random coins honestly (\ie uniformly at random).
			
			\item To each party in $\bigcup_{j \in \set{1,\ldots,d}} \cL_j^\cP$: send a message according to input $ 1-v_i$.
			
			\item To each party in $\bigcup_{j \in \set{d+1 ,\ldots, h}} \cL_j^\cP$: send a message according to input $v_i$ (real input).
			
			\item Send no messages to the other parties in $\cP$.
		\end{enumerate}
		
		\item[Other parties.] Act according to $\Pi$.
	\end{description}
	
	\item [Second round:]~
	\begin{description}
		\item[Parties in $\cP \setminus \cs$.] If $d=0$ [\resp $d=h$], act honestly according to $\Pi$ \wrt input bit $v_i$ [\resp $1-v_i$]; otherwise, abort.
		
		\item[Parties in $\cS$.] Abort.

		\item[Other parties.] Act according to $\Pi$.
	\end{description}
\end{description}
\end{protocol}
}
Namely, the ``pivot'' parties in $\cP$ shift their inputs from their real input to the flipped one according to parameter $d$. The ``aborting'' parties in $\cs$ abort at the end of the first round. Note that protocol $\PP_{0}$ is the same as protocol $\Pi^\cs$, and $\PP_h(\vv)$ acts like $\Pi^\cs(\vv')$, for $\vv'$ being $\vv$ with the coordinates in $\cP$ flipped.

For $\cP, \cS \subseteq [n]$, let $\ops = [n] \setminus (\cP \cup \cs)$, let $d\in (h)$, let $c\in \zo$, and let
\[
\VV_{d,c} = \set{(\vf,\cs,\vr) \colon \quad \PP_d(\vv;(\vf,\vr))_\ops = c^{\size{\ops}}}.
\]
Namely, $\VV_{d,c}$ are the sets, setup parameters and random strings on which honest parties in $\PP_{d}$ halt in the second round and output $c$. Let $\chi = \pr{ \Pi(\vv) \ne \o^n}$ and let
\[
\cT_{d,c}^\cP = \set{\vf \colon \ppr{\cs \gets \bns}{\ppr{\vr \gets \rn}{(\vf,\cS,\vr),(\vf,\emptyset,\vr) \in \VV_{d,c} } \ge \lambda } \ge 1-\delta}.
\]
The proof of  \cref{lemma:SecondRound:PR} immediately follows by the next lemma.
\begin{lemma}\label{lemma:SecondRound:PRR}
For every $k$-size subset $\cP \subset [n]$ and $d\in [h]$,
it holds that
\begin{align*}
\ppr{ \DF}{\cT_{d,\o}^\cP} \ge \gamma- 7\lambda- \frac{\chi + \alpha }{\lambda}.
\end{align*}
\end{lemma}
\begin{proof}[Proof of \cref{lemma:SecondRound:PR}]
Immediate by \cref{lemma:SecondRound:PRR}.
\end{proof}
\newcommand{\tV}{\widetilde{\cV}}
\newcommand{\tT}{{\widetilde{\cT}}}

The rest of this subsection is devoted to proving \cref{lemma:SecondRound:PRR}. Fix a $k$-size subset $\cP \subset [n]$ and omit it from the notation when clear from the context.
Let
\[
\tV_{d,c} = \set{(\vf,\cs,\vr) \colon \forall a\in \zo \quad \PP_{d+a}(\vv;(\vf,\vr))_\ops = c^{\size{\ops}}}.
\]
Namely, $\tV_{d,c} \subseteq \cV_{d,c}$ are the sets, setup parameters and random strings, on which honest parties in $\PP_{d+a}$ halt in the second round and output $c$, if the parties in $\cS$ abort \emph{and} \emph{regardless} of whether the parties in $\cP$ act toward those in $\cL_{d+1}$ according to input $0$ or $1$.
Let
\[
\tT_{d,c} = \set{\vf \colon \ppr{\cs \gets \bns}{\ppr{\vr \gets \rn}{(\vf,\cS,\vr),(\vf,\emptyset,\vr) \in \tV_{d,c} } \ge \lambda} \ge 1-\delta},
\]
let $\tT_{d} = \tT_{d,0}\cup \tT_{d,1}$, and let $\tT = \bigcap_{d \in (h-1)}\tT_{d}$. \cref{lemma:SecondRound:PRR} is proved via the following claims (the following probabilities are taken over $\vf \gets \DF$).

\begin{claim}\label{claim:PR:0}
$\pr{\cT_{d+1,\o} \mid \tT} < \eta$ implies $\pr{\cT_{d,1-\o} \mid \tT} \ge 1- \eta$.
\end{claim}
\begin{proof}[Proof of \cref{claim:PR:0}]
Assuming $\pr{\cT_{d+1,b}\mid \tT}\le \eta$ notice that
\[
\pr{\tT_{d,b}\mid \tT}\le \pr{\cT_{d+1,b}\mid \tT}\le \eta.
\]
Consequently, since $\pr{\tT_{d}\mid \tT}=1$, it follows that $\pr{\tT_{d,b}\mid \tT}\le \eta$ implies $\pr{\tT_{d,1-b}\mid \tT}\ge 1-\eta$ and thus $\pr{\cT_{d,1-b}\mid \tT}\ge \pr{\tT_{d,1-b}\mid \tT}\ge 1-\eta$.
\end{proof}

\begin{claim}\label{claim:PR:00}
$\pr{\tT} \ge \gamma- 5\lambda$.
\end{claim}

\begin{claim}\label{claim:PR:01}
$\pr{\cT_{1,\o} \mid \tT} \ge 1- (\chi + \alpha)/(\Pr[\tT]\cdot \lambda)$.
\end{claim}

\begin{claim}\label{claim:PR:1}
For every $d \in [h-1]$.
\[
\pr{\cT_{d,0} \mid \tT} + \pr{\cT_{d,1} \mid \tT} \le 1 + \frac{\lambda}{h\cdot\Pr[\tT]}.
\]
\end{claim}

We prove \cref{claim:PR:00,claim:PR:01,claim:PR:1} below, but first use the above claims for proving \cref{lemma:SecondRound:PR}.

\paragraph{Proving \cref{lemma:SecondRound:PRR}.}

\begin{proof}[Proof of \cref{lemma:SecondRound:PRR}.]
We first prove that for every $d\in [h]$:
\begin{align}\label{eq:SecondRound:PRR}
\pr{\cT_{d,\o} \mid \tT} \ge 1- \frac{ \chi + \alpha}{\Pr[\tT]\cdot \lambda} - \frac{d\lambda}{h\cdot\Pr[\tT]}
\end{align}	
The proof is by induction on $d$. The base case, $d=1$, is by \cref{claim:PR:01}. The induction steps follows by the combination of \cref{claim:PR:1} and the contrapositive of \cref{claim:PR:0}. Applying \cref{eq:SecondRound:PRR} for $d=h$, yields that
\begin{align*}
\pr{\cT_{h,\o} } \ge \Pr[\tT]- \frac{ \chi +\alpha}{\lambda} - \lambda,
\end{align*}	
and the proof follows by \cref{claim:PR:00}.
\end{proof}

\newcommand{\PPP}{\Pi^{\cS}}

So it is left to prove \cref{claim:PR:00,claim:PR:01,claim:PR:1}. Note that the following adversaries corrupt at most $k + \ell + 2\sigma n \le t$ parties and thus they make a valid attack. Since our security model considers rushing adversaries, and $\Pi$ has public randomness, we assume the adversary knows $\vf = (f_1,\ldots,f_n)$ \emph{before} sending its first-round messages. In the following we let $\PPP_d = \PP_d$ and $\Pi_d = \Pi^{\emptyset}_d$.

\paragraph{Proving \cref{claim:PR:00}.}
This is the only part in proof where we exploit the fact that the protocol is secure against \emph{adaptive} adversaries.
\begin{proof}[Proof of \cref{claim:PR:00}.]
For $d\in (h)$, let $\VV_d = \VV_{d,0} \cup \VV_{d,1} $ and $\tV_d = \tV_{d,0} \cup \tV_{d,1}$. Since $\ppr{\vr \gets \rn}{(\vf,\cS,\vr),(\vf,\emptyset,\vr) \in \tV_{d}} \le \sum_{c\in \zo}  \ppr{\vr \gets \rn}{(\vf,\cS,\vr),(\vf,\emptyset,\vr) \in \tV_{d,c}} $, for $f\notin \tT_{d}$ it holds  that
\begin{align}\label{eq:alg:PR:00:1}
\ppr{\cs \gets \bns}{\ppr{\vr \gets \rn}{(\vf,\cS,\vr),(\vf,\emptyset,\vr) \in \tV_{d} } \ge 2\lambda} < \delta
\end{align}

\noindent
Consider the following \emph{rushing adaptive} adversary.
{ \samepage
\begin{algorithm}[$\Ac$]\label{alg:PR:00}~
		
\begin{description}
	
\item[Pre interaction:]~
Corrupt the parties in $\cP$.

\item[First round.] ~
Let $\vf$ be the parties' setup parameters.

Do $\ceil{1/\lambda\delta}$ times:
\begin{enumerate}
	\item  Sample  $\cS \gets \bns $ conditioned on $ \size{\cs} \le 2\sigma n$. \label{step:alg:PR:00:1}
	\item  	For each   $i\in (h-1)$: estimate $\xi_i =  \ppr{\vr \gets \rn}{(\vf,\cs,\vr),(\vf,\emptyset,\vr)  \in \tV_{i}}$ by taking $\Theta( \log (h/ \lambda))$ samples of $\vr$. Let $\xi_i'$ be the result of this estimation.
	
	\item Let $d = \argmin_{i\in (h-1)} \set{\xi_i'}$.
			
	\item If $\xi'_d < 3\lambda$, break the loop.
\end{enumerate}
		
Corrupt the parties in $\cs \cup \cL_{d+1}$ ($\cs$ is the set sampled in the last loop), and act according to $\Pi_d$.

\item[Second round.]~

Let $\vr$ be the parties' second-round randomness.

If $(\vf,\cW,\vr) \notin \cV_{d+a}$ for some $a\in \zo$ and $\cW \in \set{\emptyset, \cs}$,

\quad act according to $\Pi^\cW_{d+a}$.

Else, abort.
\end{description}			
\end{algorithm}
}

By definition,  if the attack does not abort then it  violates either agreement or (second-round) halting. Let $D$ be the value of $d$ chosen by the adversary $\Ac$ at  the first round of the protocol. By construction, the  attack abort with probability $\xi_D$. So it is left to argue about the value of  $\xi_D$.

Assume $\vf \notin \tT$.  Recall that (by Chernoff/Hoeffding bound) $\ppr{\cs \gets \bns}{ \size{\cs} > 2\sigma n} \le \alpha < \delta/2$. Therefore, with probability at least $\delta/2$ over the choice of $\cs$ in Step \ref{step:alg:PR:00:1} of \Ac, there exists $i\in (h-1)$ such that $\xi_i < 2\lambda$ (by \cref{eq:alg:PR:00:1}). It follows that $\xi_D < 3\lambda$, except with probability at most $\lambda$ (\ie  error estimating $\xi_D$ by another Chernoff/Hoeffding bound). We conclude that if $\vf \notin \tT$ the attack succeeds with probability at least  $1-4\lambda$.

It follows that under the above attack,   the honest parties halt in the second round and output the same value  with probability at most $\Pr[\tT] + \Pr[\neg \tT]\cdot 4\lambda \le \Pr[\tT] + 4\lambda$. Since the parties halt \emph{and} agree with probability at least $\gamma-\alpha$, we conclude that $\Pr[\tT] \ge \gamma- \alpha - 4\lambda \ge \gamma- 5\lambda$.
\end{proof}

\paragraph{Proving \cref{claim:PR:01}.}

\newcommand{\oH}{{\overline{\cH}}}

\begin{proof}[Proof of \cref{claim:PR:01}.]
By definition, for $\vf \in \cT_{1,\o}$ it holds that
\begin{align*}
\ppr{\vr \gets \rn}{\Pi_1(\vv;(\vf,\vr))_\oH = \oo^{\size{\oH}} } = \ppr{\vr \gets \rn}{(\vf,\emptyset,\vr) \in \cV_{1,\oo} } \ge \lambda,
\end{align*}
letting $\cH = \cP \cup \cL_1$ and $\oH = [n] \setminus\cH$. Let $\eta = \ppr{\vf}{\cT_{1,\oo} \mid \tT}$, clearly,	$\ppr{\vf}{\cT_{1,\o} \mid \tT} = 1 -\eta$. By the above
\begin{align}
 \pr{\Pi_1(\vv)_\oH = \oo^{\size{\oH}}} \ge \Pr[\tT]\cdot \eta \cdot \lambda
\end{align}
(recall that $\Pi_1(\vv)$ stands for $\Pi_1(\vv;(\vf,\vr))$, for a random choice of $(\vf,\vr))$).
Finally, we notice that
\begin{align}
\pr{\Pi_1(\vv) = \oo^{\size{\wb{\cH}}}} + \pr{\Pi(\vv) = \o^n} \le 1+ \alpha
\end{align}
\noindent
If not, then the following attack violates the $\alpha$-agreement. Recall that $\Pi$ is an honest execution on input $\vv$ and $\Pi_1$ is an execution of the protocol where the parties in $\cL_1$ receive inputs from $\cP$ according to input $\vv'$ and all others receive inputs from $\cP$ according to input $\vv$ (recall that $\vv$ and $\vv'$ differ on exactly those indices indexed by $\cP$). The attack proceeds as follows: the adversary corrupts the parties in $\cH$, partitions the honest parties into two equal-size sets and acts toward the first honest parties according to $\Pi$ and toward the rest according to $\Pi_1$. We conclude that $\Pr[\tT]\cdot \eta \cdot \lambda \le \chi + \alpha $, and therefore $\eta \le (\chi + \alpha)/(\Pr[\tT]\cdot \lambda)$.
\end{proof}

\paragraph{Proving \cref{claim:PR:1}.}
The proof uses \cref{con:IsoBot} in a similar way to the second part of the proof of the theorem.

\begin{proof}[Proof of \cref{claim:PR:1}.]
For $\vr \in (\cR\cup\sset{\bot})^n$ let $ \cE(\vr)$ be the indices in $\vr$ of the value $\bot$.
We assume \wlg that a party aborts upon getting $\perp$ as its second-round random coins. For $\vf \in \SDF$, for $d\in [h-1]$, and for $\o \in \zo$, let
\begin{align}
\cA^\vf_\o = \set{\vr \in \rbot \colon \Pi_d(\vv;(\vf,\vr))_{\overline{\cP \cup \cL_d \cup \cE(\vr)}} = \o^{\size{\overline{\cP \cup \cL_d \cup \cE(\vr)}}}}.
\end{align}
By definition, for $\vf \in \cT_{d,0} \cap \cT_{d,1}$ and $\o\in \zo$, it holds that
\begin{align}
\ppr{\cs\gets \bns}{\ppr{\vr \gets \rn}{\vr,\bot_{\cS}(\vr) \in \cA^\vf_\o} \ge \lambda} \ge 1-\delta.
\end{align}
By \cref{con:IsoBot}, see \cref{eq:IsoBot}, for $\vf \in \cT_{d,0} \cap \cT_{d,1}$ it holds that
\begin{align*}
\ppr{\vr\gets \rn,\cS \gets \bns}{\forall \o \in \zo \colon \set{\vr,\bot_{\cS}(\vr)}\cap \cA^f_b \neq \emptyset} > \delta.
\end{align*}
That is,
\begin{align}\label{eq:PR:1:4}
\ppr{\vr\gets \cR^n,\cS \gets \bns}{\forall \o\in \zo \quad \exists \cs_\o \in \set{\cs,\emptyset} \colon \Pi^{\cS_\o}_{d}(\vv;(\vf,\vr))_{ \overline{\cP \cup \cL_d \cup\cS_\o} }= \o^{ \size{\overline{ \cP \cup \cL_d \cup\cS_\o}}}} >\delta.
\end{align}
In pursuit of contradiction, assume that $\pr{\cT_{d,0} \mid \tT} + \pr{\cT_{d,1} \mid \tT} \ge 1 + \lambda/(h\cdot \Pr[\tT])$ for some $d \in [h-1]$. It follows that
\begin{align}\label{eq:PR:1:5}
\lefteqn{\ppr{\stackrel{\vf \gets \DF}{\vr\gets \cR^n,\cS \gets \bns}}{\forall \o\in \zo \quad \exists \cs_\o \in \set{\cs,\emptyset}\colon \Pi^{ \cs_\o}_{d}(\vv;(\vf,\vr))_{ \overline{\cP \cup \cL_d \cup \cs_\o} }= \o^{ \size{\overline{ \cP \cup \cL_d \cup \cs_\o}}}}}\\
&> \pr{\cT_{d,0} \cap \cT_{d,1}}\cdot \delta \hskip30em \nonumber\\
&\ge \Pr[\tT]\cdot \Pr[\cT_{d,0} \cap \cT_{d,1} \mid \tT] \cdot \delta \nonumber\\
&\ge \Pr[\tT]\cdot \frac{\lambda}{h\cdot \Pr[\tT]} \cdot \delta \nonumber\\
&= \lambda\delta /h \nonumber\\
&> 8\alpha.\nonumber
\end{align}
The first inequality is by \cref{eq:PR:1:4}, the second one by the assumption that $\Pr[\cT_{d,0} \mid \tT] + \Pr[\cT_{d,1} \mid \tT] \ge 1 + \lambda/(h\cdot \Pr[\tT])$, and the last one by the definition of $\alpha$. Next, consider the following rushing adversary:

\begin{algorithm}[$\Ac$]~

\begin{description}
\item[Pre-interaction.] ~
		
	\begin{enumerate}
		\item For each $i\in [h-1]$, estimate
		\[
        \xi_i = \ppr{\vr\gets \cR^n,\cS \gets \bns}{\forall \o\in \zo \quad \exists \cs_\o \in \set{\cs,\emptyset} \colon \Pi^{\cs_\o}_{d}(\vv;(\vf,\vr))_{\overline{ \cP \cup \cL_d \cup \cs_\o}}= \o^{ \size{\overline{ \cP \cup \cL_d \cup \cs_\o}}}}
         \]
         by taking $\Theta(\log (h/\alpha))$ samples. Let $d = \argmax_{i\in [h-1]}\set{\xi_i}$.
		
		\item Sample a random $\cS \gets \bns $ conditioned on $\size{\cs} \le 2\sigma n$.
	\end{enumerate}
	
	 Corrupt the parties in $\cP \cup \cS\cup \cL_d$.
	
	\item[First round.] Act according to $\Pi_d$.
	
    \item[Second round.] Partition the honest parties arbitrarily into two equal-size sets $\cH_1$ and $\cH_2$, and act towards $\cH_1$ according to $\Pi^{\cs}_{d}$ and towards $\cH_2$ according to $\Pi^{\emptyset}_{d}$.
\end{description}
\end{algorithm}

Observe that \cref{eq:PR:1:5} says that with probability $8\alpha$ (over the setup parameter, the choice of set $S$ and coins $r$) the output of the honest parties is sensitive to whether the parties in S abort or not (while halting and agreement occurs for both cases). Therefore, analogously to the proof of \cref{claim:PR:00}, we deduce that the adversary described above causes disagreement with probability at least $\alpha$.
\end{proof}

\paragraph{Acknowledgements.}
We would like to thank Rotem Oshman, Juan Garay, Ehud Friedgut, and Elchanan Mossel for very helpful discussions.

\bibliographystyle{abbrvnat}
\bibliography{crypto}

\appendix
\section{Locally Consistent Security to Malicious Security}\label{sec:LocalToFull}

\newcommand{\VRF}{\ensuremath{\mbox{\footnotesize{\textsf{VRF}}}}\xspace}
\newcommand{\DSig}{\ensuremath{\mbox{\footnotesize{\textsf{DS}}}}\xspace}
\newcommand{\NIZK}{\ensuremath{\mbox{\footnotesize{\textsf{NIZK}}}}\xspace}
\newcommand{\ZKPoK}{\textsf{ZKPoK}\xspace}

\newcommand{\VRFGen}{\mathsf{VRF.Gen}}
\newcommand{\VRFEval}{\mathsf{VRF.Eval}}
\newcommand{\VRFVerify}{\mathsf{VRF.Verify}}

\newcommand{\Expt}{\ensuremath{\mathsf{Expt}}\xspace}
\newcommand{\ExptVRF}{\ensuremath{\Expt^\mathsf{VRF}}\xspace}
\newcommand{\ExptSig}{\ensuremath{\Expt^\mathsf{Sig}}\xspace}
\newcommand{\Oeval}{\ensuremath{\mathcal{O}_\mathsf{eval}}\xspace}
\newcommand{\Osign}{\ensuremath{\mathcal{O}_\mathsf{sign}}\xspace}

\newcommand{\Fzk}{\ensuremath{\mathcal{F}_\mathsf{zk}}\xspace}

\newcommand{\Min}{\ensuremath{\mathcal{M}_\mathsf{in}}\xspace}
\newcommand{\TMin}{\ensuremath{\tilde{\mathcal{M}}_\mathsf{in}}\xspace}
\newcommand{\Mout}{\ensuremath{\mathcal{M}_\mathsf{out}}\xspace}

\newcommand{\NIZKGen}{\mathsf{NIZK.Gen}}
\newcommand{\NIZKProve}{\mathsf{NIZK.Prover}}
\newcommand{\NIZKVerify}{\mathsf{NIZK.Verifier}}
\newcommand{\nizkproof}{\varphi}
\newcommand{\crs}{\ensuremath{\mathsf{crs}}\xspace}
\newcommand{\Rel}{\ensuremath{\cR}\xspace}
\newcommand{\Lang}{\ensuremath{\cL}\xspace}
\newcommand{\Simnizk}{\Sim_\mathsf{nizk}\xspace}

\newcommand{\DSGen}{\mathsf{DS.Gen}}
\newcommand{\DSSign}{\mathsf{DS.Sign}}
\newcommand{\DSVerify}{\mathsf{DS.Verify}}

\newcommand{\dssk}{\mathsf{sk}^{\mathsf{ds}}}
\newcommand{\dsvk}{\mathsf{vk}^{\mathsf{ds}}}
\newcommand{\vrfsk}{\mathsf{sk}^{\mathsf{vrf}}}
\newcommand{\vrfvk}{\mathsf{vk}^{\mathsf{vrf}}}

\newcommand{\setup}{\mathsf{setup}}
\newcommand{\comp}{\mathsf{Comp}}
\newcommand{\comppc}{\comp_\mathsf{PR}}

In this section, we formally state and prove \cref{thm:local_to_malicious} and show how to \emph{compile} any BA protocol that is secure against locally consistent adversaries into a protocol that is secure against malicious adversaries.
That is, we prove the following theorem:

\begin{theorem}[\cref{thm:local_to_malicious}, restated]\label{thm:local_to_malicious:Res}
Let $\Pi$ be a $(t,\alpha,\beta,q,\gamma)$-\BA against locally consistent adversaries for $q=O(\log{n})$ and assume the existence of verifiable random functions and existentially unforgeable digital signatures under an adaptive chosen-message attack.
Then,
\begin{enumerate}
    \item\label{thm:local_to_mal_generic}
    Assuming in addition the existence of non-interactive zero-knowledge proofs, there exist a \ppt protocol-compiler $\comp(\cdot)$ such that $\Pi'=\comp(\Pi)$ is a $(t,\alpha-\negl(\secParam),\beta-\negl(\secParam),q,\gamma-\negl(\secParam))$-\BA in the PKI model, resilient to malicious adversaries.
    \item\label{thm:local_to_mal_PC}
    There exists a \ppt protocol-compiler $\comppc(\cdot)$ such that if $\Pi$ is a public-randomness protocol, then $\Pi'=\comppc(\Pi)$ is a $(t,\alpha-\negl(\secParam),\beta-\negl(\secParam),q,\gamma-\negl(\secParam))$-\BA in the PKI model, resilient to malicious adversaries.
\end{enumerate}
\end{theorem}

In \cref{sec:ltf:prelim}, we define the cryptographic primitives used in the compiler, and in \cref{sec:ltf:compiler}, we construct the compiler and prove its security.

\subsection{Preliminaries}\label{sec:ltf:prelim}
The compiler makes use of  \emph{verifiable random functions} (VRF)~\cite{MRV99}, \emph{digital signatures}, and non-interactive zero-knowledge proofs, as defined below.

\subsubsection{Verifiable Random Functions}
We follow the definition of VRF from \cite{HJ16}.
\begin{definition}[VRF]\label{def:vrf}
A \textsf{verifiable random function} is a tuple of polynomial-time algorithms $\Pi=(\VRFGen,\VRFEval,\VRFVerify)$ of the following form.
\begin{itemize}
    \item
    $\VRFGen(1^\secParam)\to(\sk,\vk)$. On input the security parameter, the key-generation algorithm outputs a secret key $\sk$ and a public verification key $\vk$.
    \item
    $\VRFEval(\sk,x)\to(y,\pi)$. On input the secret key and an input $x\in\zo^\secParam$, the evaluation algorithm outputs a value $y\in \cS$ (for a finite set $\cS$) and a proof $\pi$.
    \item
    $\VRFVerify(\vk,x,y,\pi)\to b$. On input the verification key, an input $x\in\zo^\secParam$, an output $y\in \cS$, and a proof $\pi$, the deterministic verification algorithm outputs a bit $b\in\zo$.
\end{itemize}
\end{definition}
\noindent
We require the following properties:
\begin{itemize}
\item \textbf{Correctness.}
For $(\sk,\vk)\gets\VRFGen(1^\secParam)$ and $x\in\zo^\secParam$ it holds that if $(y,\pi)\gets\VRFEval(\sk,x)$ then $\VRFVerify(\vk,x,y,\pi)=1$.

\item \textbf{Unique provability.}
For all strings $(\sk, \vk)$ (not necessarily generated by $\VRFGen$) and all $x\in\zo^\secParam$, there  exists no $(y_0,\pi_0,y_1,\pi_1)$ such that $y_0\neq y_1$ and $\VRFVerify(\vk,x,y_0,\pi_0)=\VRFVerify(\vk,x,y_1,\pi_1)=1$.

\item \textbf{Pseudorandomness.}
For any \ppt adversary $\Adv=(\Adv_1,\Adv_2)$ it holds that
\[
\size{\pr{\ExptVRF_{\Pi,\Adv}(\secParam)=1}-\frac12}\leq \negl(\secParam),
\]
\noindent
for the experiment $\ExptVRF$ defined below:
\end{itemize}

\begin{small}
\begin{center}
\begin{tabular}{|l|l|}
    \hline
    \Centerstack{
    $\ExptVRF_{\Pi,\Adv}(\secParam)$
    }
    &
    \Centerstack{
    $\Oeval(x)$
    }\\
    \hline
    \Centerstack[l]{
    $(\sk,\vk)\gets\VRFGen(1^\secParam)$\\
    $(\xs,\state)\gets\Adv_1^{\Oeval(\cdot)}(\vk)$\\
    $(y_0,\pi)\gets\VRFEval(\sk,\xs)$\\
    $y_1\gets_R \cS$\\
    $b\gets_R \zo$\\
    $b'\gets\Adv_2^{\Oeval(\cdot)}(\state,y_b)$\\
    return $1$ if and only if $b=b'$\\
    \quad and $\Adv$ didn't query $\xs$
    }
    &
    \shortstack[l]{
    $(y,\pi)\gets\VRFEval(\sk,x)$\\
    return $(y,\pi)$
    }\\
    \hline
\end{tabular}
\end{center}
\end{small}

\subsubsection{Digital Signatures}
We consider the standard notion of existentially unforgeable signatures under an adaptive chosen-message attack~\cite{GMR88}.
\begin{definition}[Digital signatures]
A \textsf{digital signatures} scheme is a tuple of polynomial-time algorithms $\Pi=(\DSGen,\DSSign,\DSVerify)$ of the following form.
\begin{itemize}
    \item
    $\DSGen(1^\secParam)\to(\sk,\vk)$. On input the security parameter, the key-generation algorithm outputs a secret signing key $\sk$ and a public verification key $\vk$.
    \item
    $\DSSign(\sk,m)\to\sigma$. On input the signing key and a message $m$, the signing algorithm outputs a signature $\sigma$.
    \item
    $\DSVerify(\vk,m,\sigma)\to b$. On input the verification key, a message $m$, and a signature $\sigma$, the deterministic verification algorithm outputs a bit $b\in\zo$.
\end{itemize}
\end{definition}
\noindent
We require the following properties:
\begin{itemize}
\item \textbf{Correctness.}
For $(\sk,\vk)\gets\DSGen(1^\secParam)$ and a message $m$ it holds that if $\sigma\gets\DSSign(\sk,m)$ then $\DSVerify(\vk,m,\sigma)=1$.

\item \textbf{Existentially unforgeable under an adaptive chosen-message attack.}
For any \ppt adversary $\Adv$ it holds that
\[
\size{\pr{\ExptSig_{\Pi,\Adv}(\secParam)=1}}\leq \negl(\secParam),
\]
\noindent
for the experiment $\ExptSig$ defined below:
\end{itemize}

\begin{small}
\begin{center}
\begin{tabular}{|l|l|}
    \hline
    \Centerstack{
    $\ExptSig_{\Pi,\Adv}(\secParam)$
    }
    &
    \Centerstack{
    $\Osign(m)$
    }\\
    \hline
    \Centerstack[l]{
    $(\sk,\vk)\gets\DSGen(1^\secParam)$\\
    $(m,\sigma)\gets\Adv^{\Osign(\cdot)}(\vk)$\\
    return $1$ if and only if $\DSVerify(\vk,m,\sigma)=1$\\
    \quad and $\Adv$ didn't query $m$
    }
    &
    \shortstack[l]{
    $\sigma\gets\DSSign(\sk,m)$\\
    return $\sigma$
    }\\
    \hline
\end{tabular}
\end{center}
\end{small}

\subsubsection{Non-Interactive Zero-Knowledge Proofs}
A non-interactive zero-knowledge proof~\cite{BFM88} is a single-message protocol that allow a prover to convince a verifier the a certain common statement belongs to a language, without disclosing any additional information. We follow the definition from \cite{GOS12}.
\begin{definition}[NIZK]\label{def:nizk}
Let $\Rel$ be an $\NP$-relation and let $\Lang_\Rel$ be the language consisting of the statements in $\Rel$.
A \textsf{non-interactive zero-knowledge proof} system for $\Rel$ is a tuple of polynomial-time algorithms $\Pi=(\NIZKGen,\NIZKProve,\NIZKVerify)$ of the following form:
\begin{itemize}
    \item
    $\NIZKGen(1^\secParam)\to\crs$. On input the security parameter, the setup-generation algorithm outputs a common reference string $\crs$.
    \item
    $\NIZKProve(\crs,x,w)\to\nizkproof$. On input the $\crs$, a statement $x$, and a witness $w$ such that $(x,w)\in\Rel$, the prover algorithm outputs a proof string $\nizkproof$.
    \item
    $\NIZKVerify(\crs,x,\nizkproof)\to b$. On input the $\crs$, a statement $x$, and a proof $\nizkproof$, the verification algorithm outputs a bit $b\in\zo$.
\end{itemize}
\end{definition}
\noindent
We require the following properties:
\begin{itemize}
\item \textbf{Correctness.}
A proof system is complete if an honest prover with a valid witness can convince an honest verifier. For $(x,w)\in\Rel$ it holds that
\[
\pr{\NIZKVerify(\crs,x,\nizkproof)=1 \mid \crs\gets\NIZKGen(1^\secParam), \nizkproof\gets\NIZKProve(\crs,x,w)}=1.
\]
\item \textbf{Statistical soundness.}
A proof system is sound if it is infeasible to convince an honest verifier when the statement is false.
For all polynomial-size families $\sset{x_\secParam}$ of statements $x_\secParam\notin\Lang_\Rel$ and all adversaries $\Adv$ it holds that
\[
\pr{\NIZKVerify(\crs,x_\secParam,\nizkproof)=1 \mid \crs\gets\NIZKGen(1^\secParam), \nizkproof\gets\Adv(\crs,x_\secParam)}=1.
\]
\item \textbf{Computational (adaptive, multi-theorem) zero knowledge.}
A proof system is zero-knowledge if the proofs do not reveal any information about the witnesses. There exists a polynomial-time simulator $\Simnizk=(\Simnizk^1,\Simnizk^2)$, where $\Simnizk^1$ returns a simulated \crs together with a simulation trapdoor $\tau$ that enables $\Simnizk^2$ to simulate proofs without having access to the witness. That is, for every non-uniform polynomial-time adversary $\Adv$ it holds that
{\small{
\[
\abs{\pr{\Adv^{\Pc_\crs(\cdot,\cdot)}(\crs)= 1 \mid \crs\gets\NIZKGen(1^\secParam)}
-\pr{\Adv^{\Sim_{\crs,\tau}(\cdot,\cdot)}(\crs)= 1 \mid (\crs,\tau)\gets\Simnizk^1(1^\secParam)}}\leq\negl(\secParam),
\]
}}
where $\Sim_{\crs,\tau}(x,w)=\Simnizk^2(\crs,\tau,x)$ for $(x,w)\in\Rel$ and $\Pc_\crs(x,w)=\NIZKProve(\crs,x,w)$.
\end{itemize}

\subsubsection{Next-Message Functions}
An $n$-party protocol is represented by a set $\sset{\nextmsg_{i\to j}}_{i,j\in[n]}$ of next-message functions, a set $\sset{\outputf_i}_{i\in[n]}$ of output functions, and a distribution $D$ for generating setup information. Initially, the setup information is sampled as $(\setup_1,\ldots,\setup_n)\gets D$ and every party $\Party_i$ receives $\setup_i$ before the protocol begins.
The view of a party $\Party_i$ in the \rth round, denoted $\view_i^r$, consists of: its input bit $x_i$, its setup information $\setup_i$, its random coin tosses $\rho_i=(\rho_i^1,\ldots,\rho_i^r)$ (where $\rho_i^{r'}$ are the tossed coins for round $r'$) and the incoming messages $(m^{r'}_{1\to i}, \ldots, m^{r'}_{n\to i})$ for every $r'<r$, where $m^{r'}_{j\to i}$ is the message received from $\Party_j$ in round $r'$.
Given $\Party_i$'s view in the \rth round, the function $\nextmsg_{i\to j}(\view_i^r)$ outputs the message $m^r_{i\to j}$ to be sent by $\Party_i$ to $\Party_j$, except for the last round, where it outputs $\bot$; in that case the output function $\outputf(\view_i^r)$ produces the output value $y$. \Wlg we assume that a message $m^r_{i\to j}$ is of the form $(r,i,j,m)$; looking ahead, this will ensure that two messages in the protocol will not have the same signature.

\subsubsection{The PKI Model}
The compiled protocol is designed to work in the \emph{public-key infrastructure} (PKI) model, where a trusted third party generates private/public keys for the parties before the protocol begins. In our setting, we will require a PKI for VRF, digital signatures, and \NIZK, meaning that the trusted party operates as follows:
\begin{enumerate}
    \item
    For every $i\in[n]$, compute VRF keys $(\vrfsk_i,\vrfvk_i)\gets\VRFGen(1^\secParam)$.
    \item
    For every $i\in[n]$, compute signature keys $(\dssk_i,\dsvk_i)\gets\DSGen(1^\secParam)$.
    \item
    Compute $\crs\gets\NIZKGen(1^\secParam)$.
    \item
    Send to every party $\Party_i$ the secret keys $(\vrfsk_i,\dssk_i)$ as well as all the public keys $\crs$, $(\vrfvk_1,\ldots,\vrfvk_n)$ and $(\dsvk_1,\ldots,\dsvk_n)$.
\end{enumerate}

\subsection{The Compiler}\label{sec:ltf:compiler}

Given a protocol that is secure against locally consistent adversaries, the main idea of the compiler is to limit the capabilities of a malicious adversary attacking the compiled protocol to those of a locally consistent one. This is achieved by proving an honest behavior via the cryptographic tools described above (VRF, digital signatures, and NIZK proofs) in a similar way to the GMW compiler~\cite{GMW87}. Unlike GMW, where all consistency proofs are carried out over a broadcast channel to ensure a consistent view between the honest parties, in our case the consistency proofs are done over pairwise channels, so they only guarantee local consistency.

\smallskip
We start by defining the NP relations that will be used for the zero-knowledge proofs. Each instance consists of a message between a pair of parties (say from $\Party_i'$ to $\Party_j'$) and the witness is the internal state of $\Party_i'$ used to generate the message (the input, the random coins, and all incoming messages) along with a ``proof of correctness,'' \ie that the random coins were properly generated using the VRF, that the incoming messages that $\Party_i'$ received from every $\Party_k'$ were signed by $\Party_k'$, and in turn were proven to be generated correctly (\ie that each $\Party_k'$ used the correct random coins generated by the VRF and its incoming messages were signed by the senders). Note that this recursive step in the verification is required for proving locally consistent behaviour, since if both $\Party_i'$ and $\Party_k'$ are corrupt, then $\Party_k'$ can send an arbitrary message to $\Party_i'$ and sign it (in this case the NIZK proof from $\Party_k'$ to $\Party_i'$ will not verify). When $\Party_i'$ sends its message to an honest $\Party_j'$, it is not enough that $\Party_i'$ proves that the messages from $\Party_k'$ are properly signed, but $\Party_i'$ must also prove that $\Party_k'$ provided a NIZK proof asserting that its messages were generated by consistent random coins and correct incoming messages according to the next-message function.
For this reason we consider $q=O(\log{n})$

\paragraph{The Relation $\Rel^r_{i\to j}$.}
We will consider the following set of NP relations, where for $i,j\in[n]$ and an integer $r$, the relation $\Rel^r_{i\to j}$ is parametrized by an $n$-party protocol $\Pi$ (represented by $\sset{\nextmsg_{i\to j}}_{i,j\in[n]}$ and $\sset{\outputf_i}_{i\in[n]}$), a $\VRF$ scheme, a $\DSig$ scheme, and a \NIZK scheme, as well as:
\begin{itemize}
	\item
	A vector of VRF verification keys $(\vrfvk_1,\ldots,\vrfvk_n)$.
	\item
	A vector of signature verification keys $(\dsvk_1,\ldots,\dsvk_n)$.
	\item
	A \NIZK common reference string $\crs$.
\end{itemize}
The instance consists of a message $(m_{i\to j}^r,\sigma_{i\to j}^r,\pi_i^r)$ (the message from $\Party_i$ to $\Party_j$).
The witness consists of:
\begin{itemize}
	\item
	A bit $x_i\in\zo$ and a string $\setup_i$.
	\item
	A vector of random coins $(\rho_i^1,\ldots,\rho_i^r)$.
	\item
	For $r'\in[r-1]$ and $k\in[n]$, a message $\vm^{r'}_{k\to i}=(m^{r'}_{k\to i},\sigma^{r'}_{k\to i},\pi^{r'}_k,\nizkproof_{k\to i}^{r'})$ ($\Party_i$'s incoming messages).
\end{itemize}
The instance/witness pair is in the relation $\Rel^r_{i\to j}$ if the following holds:
\begin{enumerate}
	\item
	For every $r'\in[r]$ it holds that $\VRFVerify(\vrfvk_i,(i,r'),\rho_i^{r'},\pi_i^{r'})=1$.
	\item
	$\DSVerify(\dsvk_i,m^r_{i\to j},\sigma^r_{i\to j})=1$.
	\item
	For $r'\in[r-1]$ and $k\in[n]$ it holds that $\NIZKVerify(\crs,(m^{r'}_{k\to i},\sigma^{r'}_{k\to i},\pi_k^{r'}),\nizkproof_{k\to i}^{r'})=1$ \wrt the relation $\Rel^{r'}_{k\to i}$.
	\item
	Set $\view_i^1=(x_i,\setup_i,\rho_i^1)$ and for $1<r'\leq r$ set $\view_i^{r'}=(\view_i^{r'-1},m_{1\to i}^{r'-1},\ldots,m_{n\to i}^{r'-1},\rho_i^{r'})$. Then, it holds that $m_{i\to j}^r=\nextmsg_{i\to j}(\view_i^r)$.
\end{enumerate}

\paragraph{The compiled protocol.}
Having defined the relations $\sset{\Rel^r_{i\to j}}$, we are ready to present the compiler for a protocol $\Pi$, secure against locally consistent adversaries to a maliciously secure one. Initially, in the setup phase, each party receives its setup information for $\Pi$ in addition to the PKI keys for $\VRF$, digital signatures, and NIZK (as described above). To generate its coins for the \rth round (along with a proof), party $\Party_i$ evaluates the $\VRF$ over the pair $(i,r)$; next, $\Party_i$ computes the \rth round messages for $\Pi$, signs each message, and sends to every other $\Party_j$ the corresponding message, the signature, and the $\VRF$ proof. In addition, $\Party_i$ sends to $\Party_j$ a $\NIZK$ proof for $\Rel^r_{i\to j}$, proving that $\Party_i$ behaves consistently towards $\Party_j$.

Let $\Pi=(\Party_1,\ldots,\Party_n)$ be an $n$-party protocol represented by the set of next-message functions $\sset{\nextmsg_{i\to j}}_{i,j\in[n]}$, the set of output functions $\sset{\outputf_i}_{i\in[n]}$, and a distribution $D$ for generating setup information. Let $\VRF$ be a verifiable random function, let $\DSig$ be a digital signatures scheme, and let $\NIZK$ be a non-interactive zero-knowledge proof scheme.
Later on, we will simplify the compiler for the case of public-randomness protocols by removing the need for $\NIZK$.

\begin{protocol}[Protocol $\Pi' = (\Party'_1,\ldots, \Party'_n)=\comp(\Pi)$]\label{prot:generic}~
\begin{itemize}
\item[Setup:]
The setup-generation algorithm samples $(\setup_1,\ldots,\setup_n)\gets D$ for the protocol $\Pi$, computes $\crs\gets\NIZKGen(1^\secParam)$, and for every $i\in[n]$ computes $(\vrfsk_i,\vrfvk_i)\gets\VRFGen(1^\secParam)$ and $(\dssk_i,\dsvk_i)\gets\DSGen(1^\secParam)$.
The setup string for party $\Party'_i$ is set to be $\setup'_i=\left(\setup_i, \vrfsk_i,\dssk_i, \crs,\vrfvk_1,\ldots,\vrfvk_n, \dsvk_1,\ldots,\dsvk_n\right)$.

\item[Input:] Party $\Party_i'$  starts with an input bit $x_i\in\zo$.

\item[Round $r=1$:]~

\begin{enumerate}
	\item $\Party_i'$ computes $(\rho_i^1,\pi_i^1)\gets\VRFEval(\vrfsk_i,(i,1))$ and sets $\view_i^1=(x_i,\setup_i,\rho_i^1)$.
	\item $\Party_i'$ computes for every $j\in[n]$ the message $m^1_{i\to j}=\nextmsg_{i\to j}(\view_i^1)$ and signs $\sigma_{i\to j}^1\gets\DSSign(\dssk_i,m_{i\to j}^1)$.
	
	\item $\Party_i'$ computes for every $j\in[n]$ a proof for the relation $\Rel^1_{i\to j}$ on $\stat_{i\to j}^1=(m_{i\to j}^1,\sigma_{i\to j}^1,\pi_i^1)$ and witness $\wit_{i\to j}^1=(x_i,\setup_i,\rho^1_i)$ as $\nizkproof_{i\to j}^1\gets\NIZKProve(\crs, \stat_{i\to j}^1, \wit_{i\to j}^1)$.
	\item $\Party_i'$ sends $\vm_{i\to j}^1=(m_{i\to j}^1,\sigma_{i\to j}^1, \pi_i^1,\nizkproof_{i\to j}^1)$ to $\Party'_j$.
\end{enumerate}

\item[Round $r>1$:]
Let $\vm_{j\to i}^{r-1}=(m_{j\to i}^{r-1},\sigma_{j\to i}^{r-1}, \pi_j^{r-1},\nizkproof_{i\to j}^{r-1})$ be the message $\Party_i'$ received from $\Party_j'$ in round $r-1$. If $\Party_j'$ did not send a message, or if $\NIZKVerify(\crs,(m_{j\to i}^{r-1},\sigma_{j\to i}^{r-1}, \pi_j^{r-1}), \nizkproof_{i\to j}^{r-1})=0$, set $m_{j\to i}^{r-1}=\bot$.

\begin{enumerate}
	\item $\Party_i'$ computes $(\rho_i^r,\pi_i^r)\gets\VRFEval(\vrfsk_i,(i,r))$ and sets the internal view as $\view_i^r=(\view_i^{r-1},m_{1\to i}^{r-1},\ldots,m_{n\to i}^{r-1},\rho_i^r)$.
    \item
    $\Party_i'$ computes for every $j\in[n]$ the message $m^r_{i\to j}=\nextmsg_{i\to j}(\view_i^r)$ and signs $\sigma_{i\to j}^r\gets\DSSign(\dssk_i,m_{i\to j}^r)$.
	
	\item $\Party_i'$ computes for every $j\in[n]$ a proof for the relation $\Rel^r_{i\to j}$ on the statement $\stat_{i\to j}^r=(m_{i\to j}^r,\sigma_{i\to j}^r,\pi_i^r)$ and witness $\wit_{i\to j}^r=(\wit_{i\to j}^{r-1}, \rho^r_i, \sset{\vm_{k\to i}^{r-1}}_{k\in[n]})$ as $\nizkproof_{i\to j}^r\gets\NIZKProve(\crs, \stat_{i\to j}^r, \wit_{i\to j}^r)$.
	\item $\Party_i'$ sends $\vm_{i\to j}^r=(m_{i\to j}^r,\sigma_{i\to j}^r,\nizkproof_{i\to j}^r,\pi_i^r)$ to $\Party'_j$.
\end{enumerate}
     	
\item[Output:]
If in some round $r$, the output of $\nextmsg_{i\to j}(\view_i^r)$ is $\bot$ for all $j\in[n]$, indicating it is the last round, $\Party_i'$ outputs $y=\outputf(\view_i^r)$ and halts.
\end{itemize}
\end{protocol}

\subsubsection{Security Proof}
We prove the security of \cref{prot:generic} using a sequence of arguments. Given a protocol $\Pi$ secure against locally consistent adversaries, we first adjust it to use pseudorandom coins computed using a VRF. The new protocol, denoted $\Pi_1$, remains secure against slightly weaker locally consistent adversaries by the pseudorandomness property of the VRF. Next, we show how to convert any malicious adversary against the compiled protocol $\Pi'=\comp(\Pi)$ into a ``weak'' locally consistent attack against $\Pi_1$. The proof of the second part of the theorem, concerning public-randomness protocols, follows in similar lines.

\begin{proof}[Proof of \cref{thm:local_to_malicious:Res}]
We start by proving the first part of the theorem, considering generic protocols, and later focus on public-randomness protocols.

\paragraph{Proof of Item~\ref{thm:local_to_mal_generic} (generic protocols).}
We prove Item~\ref{thm:local_to_mal_generic} in two steps.
Initially, as an intermediate step, we consider a variant of $\Pi$, denoted $\Pi_1$, where the parties behave exactly as in $\Pi$ except that they use a $\VRF$ to compute their random coins for each round. Formally, $\Pi_1$ is defined in the PKI model, where, in addition to the setup information for $\Pi$, every party $\Party_i$ receives $\vrfsk_i$ and $(\vrfvk_1,\ldots,\vrfvk_n)$ for $(\vrfsk_i,\vrfvk_i)\gets\VRFGen(1^\secParam)$. During the execution of the protocol, each party $\Party_i$ evaluates $(\rho_i^r,\pi_i^r)\gets\VRFEval(\vrfsk_i,(i,r))$, sets its coins for the \rth round to $\rho_i^r$ (instead of a uniformly distributed string), and appends $\pi_i^r$ to its \rth round messages.
Note that the strings $\rho_i^r$ are deterministic, so a locally consistent adversary has the power to use arbitrary values instead. To enable a reduction to the security of $\Pi$, we will explicitly assume that corrupted parties indeed use the honestly generated pseudorandom values $\rho_i^r$ by evaluating the \VRF on $(i,r)$; we call such a locally consistent adversary \emph{\VRF-compliant}.
\begin{claim}
If $\Pi$ is a $\left(t,\alpha,\beta,q,\gamma\right)$-\BA against locally consistent adversaries, then $\Pi_1$ is a $\left(t,\alpha-\negl(\secParam),\beta-\negl(\secParam),q,\gamma-\negl(\secParam)\right)$-\BA against locally consistent  \VRF-compliant adversaries.
\end{claim}
\begin{proof}
By assumption, a corrupted $\Party_i$ uses the value $\rho_i^r$ as its random coins for the \rth round. Therefore, the only difference between $\Pi_1$ and $\Pi$ are the use of pseudorandom string instead of uniformly distributed strings. The proof follows by the pseudorandomness of the $\VRF$ scheme using a standard hybrid argument.
\end{proof}

Next, let $\Adv'$ be an adversary attacking $\Pi'=(\Party'_1,\ldots,\Party'_n)$. We will construct an adversary $\Adv$ for the protocol $\Pi_1=(\Party_1,\ldots,\Party_n)$. Let $\Simnizk=(\Simnizk^1,\Simnizk^2)$ be the simulator that is guaranteed for the NIZK scheme. The adversary $\Adv$ runs internally a copy of $\Adv'$ and proceeds as follows:
\begin{itemize}
\item
In the setup phase of $\Pi_1$, $\Adv$ receives the setup string $\left(\setup_i, \vrfsk_i,\vrfvk_1,\ldots,\vrfvk_n\right)$ (consisting of the setup for $\Pi$ and the \VRF keys). Next, \Adv samples $(\crs,\tau)\gets\Simnizk^1(1^\secParam)$ and $(\dssk_i,\dsvk_i)\gets\DSGen(1^\secParam)$ for every $i\in[n]$, and provides the setup string $\setup'_i=\left(\setup_i, \vrfsk_i,\dssk_i, \crs,\vrfvk_1,\ldots,\vrfvk_n, \dsvk_1,\ldots,\dsvk_n\right)$ for every corrupted $\Party'_i$.

\item
Upon receiving a message $(m^r_{i \to j},\pi_i^r)$ from an honest $\Party_i$ to a corrupted $\Party_j$ in the execution of $\Pi_1$, \Adv sends $(m^r_{i \to j},\sigma^r_{i \to j},\pi_i^r, \nizkproof_{i\to j}^r)$ to $\Adv'$ with $\sigma_{i\to j}^r\gets\DSSign(\dssk_i,m^r_{i \to j})$ and $\nizkproof_{i\to j}^r\gets\Simnizk^2(\crs,\tau,(m^r_{i \to j},\sigma^r_{i \to j},\pi_i^r))$.
\item
When $\Adv$ receives $(m^r_{i \to j},\sigma^r_{i \to j},\pi_i^r,\nizkproof_{i\to j}^r)$ from $\Adv'$ on behalf of a corrupted $\Party_i'$ to an honest $\Party_j'$ (in the simulated execution of $\Pi'$), \Adv first verifies that $\NIZKVerify(\crs,(m^r_{i \to j},\sigma^r_{i \to j},\pi_i^r),\nizkproof_{i\to j}^r)=1$. If the proof is verified, \Adv sends the message $(m^r_{i \to j},\pi_i^r)$ to $\Party_j$ in the protocol $\Pi_1$; otherwise, \Adv considers $\Party_i$ as an aborting party towards $\Party_j$.
\end{itemize}

We complete the proof in a series of steps, analyzing the attack under increasingly stronger power of the adversary $\Adv'$, starting from a locally consistent \VRF-compliant attack until reaching a full blown malicious attack.
Initially, we will assume perfect security of the NIZK, and remove this restriction later on.
\begin{claim}
Consider a perfect NIZK scheme.
If $\Pi_1$ is a $\left(t,\alpha,\beta,q,\gamma\right)$-\BA against locally consistent \VRF-compliant adversaries, then $\Pi'$ is a $\left(t,\alpha,\beta,q,\gamma\right)$-\BA against locally consistent \VRF-compliant adversaries.
\end{claim}
\begin{proof}
If $\Adv'$ is a locally consistent \VRF-compliant adversary, then in particular whenever $\Adv'$ sends a message on behalf of a corrupted $\Party_i'$, he knows a witness for the \NIZK proof. Therefore, \wlg we can assume that either a corrupted $\Party_i'$ does not send a message (\ie aborts) to an honest $\Party_j'$ or that $\Party_i'$ correctly generates the \NIZK proof. In that case every locally consistent \VRF-compliant attack by $\Adv'$ translates to a locally consistent \VRF-compliant attack by $\Adv$.
\end{proof}

The next claim considers stronger adversaries that are allowed to use arbitrary random coins for computing the next-message function.
We will use the following notations: A message sent in $\Pi'$ is of the form $(m,\sigma,\pi,\nizkproof)$; we call $m$ the \emph{content of the message}. For a party $\Party_i'$, let $\Min^{r',k\to i}$ denote the set of incoming messages' contents received from party $\Party_k'$ in round $r'$ (as this is a locally consistent attack, there could be multiple incoming messages from each corrupted party, but at most one message from each honest party). Let $\Mout^{r,i\to j}$ be the set of possible messages' contents that $\Party_i'$ can send to $\Party_j'$ at round $r$ under a \VRF-compliant locally consistent attack when using a subset of the incoming messages' contents $\sset{\Min^{r',k\to i}}_{r'<r, k\in[n]}$ and randomness $\sset{\rho_i^{r'}}_{r'\in[r]}$ computed as $(\rho_i^{r'},\pi_i^{r'})\gets\VRFEval(\vrfsk_i,(i,r'))$.

\begin{claim}\label{claim:LC:rand}
Consider a perfect NIZK scheme.
If $\Pi_1$ is a $\left(t,\alpha,\beta,q,\gamma\right)$-\BA against locally consistent \VRF-compliant adversaries, then $\Pi'$ is a $\left(t,\alpha-\negl(\secParam),\beta-\negl(\secParam),q,\gamma-\negl(\secParam)\right)$-\BA against locally consistent adversaries.
\end{claim}
\begin{proof}
We prove the claim by showing that the additional power of the adversary only allows for a negligible cheating advantage.
Consider a locally consistent adversary $\Adv'$ and assume that a corrupted party $\Party_i'$ used arbitrary random coins to generate the message content for party $\Party_j'$ in round $r$, denoted $\tilde{m}_{i\to j}^r$. There are two possible cases:
\begin{itemize}
    \item[\textbf{Case 1:}]
    If $\tilde{m}_{i\to j}^r\in\Mout^{r,i\to j}$, then the adversary can compute a witness for the relation $\Rel^r_{i\to j}$. That is, even if the actual coins used to generate $\tilde{m}_{i\to j}^r$ are different than $\sset{\rho_i^{r'}}_{r'\in[r]}$, the message $\tilde{m}_{i\to j}^r$ can be explained as if generated using $\sset{\rho_i^{r'}}_{r'\in[r]}$ consistently with a subset of the incoming messages in $\sset{\Min^{r',k\to i}}_{r'<r, k\in[n]}$. Therefore, \wlg this can be cast as a locally consistent \VRF-compliant attack.
    \item[\textbf{Case 2:}]
    If $\tilde{m}_{i\to j}^r\notin\Mout^{r,i\to j}$, let $\sset{\tilde{\rho}_i^{r'}}_{r'\in[r]}$ be the coins used by $\Adv'$ to generate $\tilde{m}_{i\to j}^r$. Then, $\tilde{\rho}_i^{r'}\neq \rho_i^{r'}$ for at least one $r'$. To provide a witness for the relation $\Rel^r_{i\to j}$, $\Adv'$ must generate $\tilde{\pi}_i^{r'}$ such that $\VRFVerify(\vrfvk_i,(i,r'),\tilde{\rho}_i^{r'},\tilde{\pi}_i^{r'})=1$. By unique provability property of the VRF, such an attack can only succeed with negligible probability.
\qedhere
\end{itemize}
\end{proof}

The next claim considers stronger adversaries that are allowed to use arbitrary incoming messages for their next-message function.
\begin{claim}\label{claim:LC:messages}
Consider a perfect NIZK scheme.
If $\Pi_1$ is a $\left(t,\alpha,\beta,q,\gamma\right)$-\BA against locally consistent \VRF-compliant adversaries, then $\Pi'$ is a $\left(t,\alpha-\negl(\secParam),\beta-\negl(\secParam),q,\gamma-\negl(\secParam)\right)$-\BA against locally consistent adversaries that are allowed to use arbitrary messages' contents when computing the next-message function.
\end{claim}
\begin{proof}
Consider an adversary $\Adv'$ that behaves locally consistent but can use arbitrary values as incoming messages. Assume that $\Adv'$ is \VRF-compliant and let $r$ be the first round in which $\Adv'$ deviates from the protocol \wrt incoming messages. Let $\Party_i'$ be a corrupted party that uses $\sset{\TMin^{r',k\to i}}_{r'<r, k\in[n]}$ as its set of incoming messages to generate the message content for party $\Party_j'$ in round $r$, denoted $\tilde{m}_{i\to j}^r$, and assume that $\bigcup \TMin^{r',k\to i} \nsubseteq \bigcup \Min^{r',k\to i}$.
There are two possible cases:

\begin{itemize}
    \item[\textbf{Case 1:}]
    If $\tilde{m}_{i\to j}^r\in\Mout^{r,i\to j}$, then the adversary can compute a witness for the relation $\Rel^r_{i\to j}$. That is, even if $\bigcup \TMin^{r',k\to i} \nsubseteq \bigcup \Min^{r',k\to i}$, the message $\tilde{m}_{i\to j}^r$ can be explained as if generated using a subset of $\bigcup \Min^{r',k\to i}$. Therefore, \wlg this can be cast as a locally consistent attack.
    \item[\textbf{Case 2:}]
    If $\tilde{m}_{i\to j}^r\notin\Mout^{r,i\to j}$, then to find a witness for the relation $\Rel^r_{i\to j}$, $\Adv'$ must produce for every message $\tilde{m}_{k\to i}^{r'}\in\bigcup \TMin^{r',k\to i} \setminus\bigcup \Min^{r',k\to i}$ a signature $\tilde{\sigma}_{k\to i}^{r'}$, a VRF proof $\pi_{k\to i}^{r'}$ and a NIZK proof $\tilde{\nizkproof}_{k\to i}^{r'}$ .
    \begin{itemize}
        \item
        If $\Party_k$ is honest, $\Adv'$ can find an accepting signature $\tilde{\sigma}_{k\to i}^{r'}$ for $\tilde{m}_{k\to i}^{r'}$ under $\dsvk_k$ only with negligible probability (recall that every message $\tilde{m}_{k\to i}^{r'}$ encodes the values $k,i,r'$; hence, $\Adv'$ cannot reuse messages that were signed by $\Party_k$ in other rounds).
        \item
        If $\Party_k$ is corrupted, then in turn it must have provided a valid witness for the relation $\Rel_{k\to i}^{r'}$. By the minimality of $r$, it is guaranteed that $\tilde{m}_{k\to i}^{r'-1}$ was honestly generated \wrt the incoming messages of $\Party_k'$ until round $r'-1$, $\sset{\bigcup \Min^{r'',k'\to k}}_{r''\in[r'-1], k'\in[n]}$. In this case, \wlg, the message $\tilde{m}_{k\to i}^{r'}$ could have been sent by the corrupted $\Party_k'$ to the corrupted $\Party_i'$, \ie be included in the set $\Min^{r',k\to i}$.
    \end{itemize}
\end{itemize}
The proof of the claim now reduces considering non-\VRF-compliant adversaries, which follows from \cref{claim:LC:rand}.
\end{proof}

The next claim considers stronger adversaries that are not required to compute their outgoing messages by the next-message function, but can send arbitrary messages instead.
\begin{claim}
Consider a perfect NIZK scheme.
If $\Pi_1$ is a $\left(t,\alpha,\beta,q,\gamma\right)$-\BA against locally consistent \VRF-compliant adversaries, then $\Pi'$ is a $\left(t,\alpha-\negl(\secParam),\beta-\negl(\secParam),q,\gamma-\negl(\secParam)\right)$-\BA against malicious adversaries.
\end{claim}
\begin{proof}
Consider a malicious adversary $\Adv'$ and assume that $\Adv'$ behaves locally consistent and \VRF-compliant until round $r$, \ie round $r$ is the first round in which $\Adv'$ does not compute a message according to the next-message function. Let $\Party_i'$ be a corrupted party that generates the message content for party $\Party_j'$ in round $r$, denoted $\tilde{m}_{i\to j}^r$, arbitrarily.
There are two possible cases:
\begin{itemize}
    \item[\textbf{Case 1:}]
    If $\tilde{m}_{i\to j}^r\in\Mout^{r,i\to j}$, then the adversary can compute a witness for the relation $\Rel^r_{i\to j}$. That is, the message $\tilde{m}_{i\to j}^r$ can be explained as if generated using $\sset{\rho_i^{r'}}_{r'\in[r]}$ consistently with a subset of the incoming messages in $\sset{\Min^{r',k\to i}}_{r'<r, k\in[n]}$ according to the next-message function. Therefore, \wlg this can be cast as a locally consistent \VRF-compliant attack.
    \item[\textbf{Case 2:}]
    If $\tilde{m}_{i\to j}^r\notin\Mout^{r,i\to j}$, then $\Adv'$ must provide $\tilde{\sigma}_{i\to j}^r$ and $\pi_i^r$ along with a witness $\wit_{i\to j}^r$ consisting of:
    \begin{itemize}
        \item
        An input bit $x_i$ an $\setup_i$.
        \item
        For every $r'\in[r]$ random coins $\rho_i^{r'}$.
        \item
        For every $r'\in[r-1]$ and $k\in[n]$ a message $\tilde{\vm}_{k\to i}^{r'}=(\tilde{m}_{k\to i}^{r'},\tilde{\sigma}_{k\to i}^{r'},\pi_k^{r'},\tilde{\nizkproof}_{k\to i}^{r'})$.
    \end{itemize}
    In addition it holds that $((\tilde{m}_{i\to j}^r,\tilde{\sigma}_{i\to j}^r,\pi_i^r),\wit_{i\to j}^r)\in\Rel^r_{i\to j}$.
    As before, with all but negligible probability it is guaranteed that $\VRFVerify(\vrfvk_i,(i,r),\rho_i^r,\pi_i^r)=1$ and for every honest party $\Party_k'$, $((\tilde{m}_{k\to i}^{r'},\tilde{\sigma}_{k\to i}^{r'},\pi_k^{r'}),\tilde{\nizkproof}_{k\to i}^{r'})\in\Rel_{k\to i}^{r'}$. For a corrupted $\Party_k'$, if $((\tilde{m}_{k\to i}^{r'},\tilde{\sigma}_{k\to i}^{r'},\pi_k^{r'}),\tilde{\nizkproof}_{k\to i}^{r'})\in\Rel_{k\to i}^{r'}$ then \wlg the message could have been sent by $\Party_k'$ to $\Party_i'$. We conclude that with all but negligible probability, the $\tilde{m}_{i\to j}^r$ can be explained by a locally consistent \VRF-compliant attack.
\end{itemize}
The proof of the claim now follows from \cref{claim:LC:messages}.
\end{proof}

Finally, we remove the assumption of a perfect NIZK scheme and consider a NIZK scheme that allows for negligible adversarial advantage, and obtain the following claim.
\begin{claim}
If $\Pi_1$ is a $\left(t,\alpha,\beta,q,\gamma\right)$-\BA against locally consistent \VRF-compliant adversaries, then $\Pi'$ is a $\left(t,\alpha-\negl(\secParam),\beta-\negl(\secParam),q,\gamma-\negl(\secParam)\right)$-\BA against malicious adversaries.
\end{claim}
This concludes the proof of the first part of the theorem.

\paragraph{Proof of Item~\ref{thm:local_to_mal_PC} (public-randomness protocols).}
We prove Item~\ref{thm:local_to_mal_PC} of \cref{thm:local_to_malicious:Res} by adjusting the compiler $\comp$ and removing the use of NIZK proofs. The new compiler $\comppc$ is defined like $\comp$ except that instead of computing a NIZK proof $\nizkproof_{i\to j}^r\gets\NIZKProve(\crs, \stat_{i\to j}^r, \wit_{i\to j}^r)$ for the relation $\Rel_{i\to j}^r$ and sending $\nizkproof_{i\to j}^r$, the sender $\Party_i'$ simply sends the witness $\wit_{i\to j}^r$. The receiver $\Party_j'$ can now directly verify that $\wit_{i\to j}^r$ is a valid witness. The proof follows immediately from Item~\ref{thm:local_to_mal_generic} of \cref{thm:local_to_malicious:Res}.
\end{proof}

\end{document}